%% file: main.tex
\documentclass[manuscript,dvipsnames,acmsmall,screen,authorversion]{acmart}\acmDOI{10.1145/3763073}
\acmYear{2025}
\acmJournal{PACMPL}
\acmVolume{9}
\acmNumber{OOPSLA2}
\acmArticle{295}
\acmMonth{10}
\usepackage{booktabs}   %% For formal tables:
\usepackage{subcaption} %% For complex figures with subfigures/subcaptions
\usepackage{amsthm}
\usepackage{mathtools}
\usepackage{stmaryrd}
\usepackage{dsfont}

\usepackage[plain]{algorithm}
\usepackage[noend]{algpseudocode}
\usepackage{tikz}
\usetikzlibrary{arrows.meta, matrix}
\usepackage[many]{tcolorbox}
\tcbuselibrary{skins}
\tcbuselibrary{listings}
\usepackage{xcolor}
\usepackage{colortbl}
\usepackage{cleveref}
\usepackage{relsize}
\usepackage{wrapfig}
\usepackage{microtype}
\usepackage{array}
\usepackage{enumitem}

\begin{document}

\title[Synthesizing DSLs for Few-Shot Learning]{Synthesizing DSLs for Few-Shot Learning}
\titlenote{Version with appendix.}

\author{Paul Krogmeier}
\orcid{0000-0002-6710-9516}
\affiliation{%
  \institution{University of Illinois at Urbana-Champaign}
  \city{Urbana}
  \country{USA}
}
\email{paulmk2@illinois.edu}

\author{P. Madhusudan}
\orcid{0000-0002-9782-721X}
\affiliation{%
  \institution{University of Illinois at Urbana-Champaign}
  \city{Urbana}
  \country{USA}
}
\email{madhu@illinois.edu}

\input{macros}
\input{abstract}

%% 2012 ACM Computing Classification System (CSS) concepts
%% Generate at 'http://dl.acm.org/ccs/ccs.cfm'.

%% End of generated code

% \keywords{domain-specific language, symbolic learning, program synthesis, tree automata}

%% \maketitle
%% Note: \maketitle command must come after title commands, author
%% commands, abstract environment, Computing Classification System
%% environment and commands, and keywords command.
\maketitle
\renewcommand{\shortauthors}{P. Krogmeier, P. Madhusudan}

\input{introduction}
\input{motivation}
\input{prelim}
\input{problem-requirements}
\input{adequate}
\input{dsl-synthesis-grammars}
\input{macrogram}
\input{related}
\input{conclusion}

\section*{Data Availability Statement}
\label{sec:data-avail-stat}

This paper has no accompanying artifact.

% Acknowledgments
\begin{acks}
  This work was supported in part by a Discovery Partners Institute
  (DPI) science team seed grant and a research grant from Amazon.
\end{acks}

%% Bibliography
\newpage

\bibliography{main}

\newif\ifappendix
\appendixtrue
% \appendixfalse

\ifappendix
\newpage
\appendix
\input{appendix}
\else
\fi

\end{document}

%% file: macros.tex
\definecolor{Lightyellow}{HTML}{E3CEAB}
\definecolor{Lightgreen}{HTML}{D3D7CF}

% tcolorbox

\newtcblisting[auto counter,number within=section,crefname={program}{programs},Crefname={Program}{Programs}]{mytcblisting}[2][]{
  enhanced,
  % breakable,
  arc=3pt,
  % outer arc=0pt,
  % toprule=0.5pt,
  % bottomrule=0.5pt,
  % leftrule=0.5pt,
  % rightrule=0.5pt,
  boxrule=0.5pt,
  left=1pt,
  right=1pt,
  listing only,
  colback=white,
  colframe=black,
  fonttitle=\normalfont,
  colbacktitle=Lightgreen,
  coltitle=black,
  % before skip=6pt,
  % after skip=6pt,
  listing options={
  },
  title={{\bfseries InterpretGrammar:} #2},
  #1
}

% end tcolorbox

% listings

\colorlet{listing-comment}{gray}
\colorlet{operator-symbol}{yellow!45!black}
\lstdefinelanguage{Haskell}{
    language=Caml,
    morekeywords={match, with, case, of, any, all},
    morekeywords=[2]{False, True},
    keywordstyle=[2]\color{black},
    morekeywords=[3]{int, bool, list, nat},
    keywordstyle=[3]\color{black},
    morekeywords={otherwise},
    literate=%
      {=}{{{\color{operator-symbol}$\mathtt{=}$}}}1
      {cup}{{{$\cup$}}}1
      {sigma}{{{$\sigma$}}}1
      {setMemb}{{$\in$}}1
      {emptySet}{{$\emptyset$}}1
      {neq}{{$\neq$}}1
      {exists}{{$\exists$}}1
      {conj}{{$\wedge$}}1
      {disj}{{$\vee$}}1
      {neg}{{$\neg$}}1
      {univ}{{$\forall x$}}1
      {A_1}{{A$_1$}}1
      {A_m}{{A$_m$}}1
      {A_i}{{A$_i$}}1
      {Q_1}{{Q$_1$}}1
      {Q_m}{{Q$_m$}}1
      {Q_i}{{Q$_i$}}1
      {while}{{{while}}}1
      {<}{{{\color{operator-symbol}<}}}1
      {>}{{{\color{operator-symbol}>}}}1
      {:}{{{\color{operator-symbol}:}}}1
      {;}{{{\color{operator-symbol};}}}1
      {|}{{{\color{operator-symbol}|}}}1
      {[}{{{\color{operator-symbol}[}}}1
      {]}{{{\color{operator-symbol}]}}}1
      {\&}{{{\color{operator-symbol}\&}}}1
      {->}{{{\color{operator-symbol}$\rightarrow$}}}1
}

\lstdefinestyle{default}{
    basicstyle=\linespread{1.0}\ttfamily,
    columns=fullflexible,
    commentstyle=\sffamily\color{black!50!white},
    escapechar=\#,
    framexleftmargin=1ex,
    framexrightmargin=1ex,
    keepspaces=true,
    breakindent=0pt,
    keywordstyle=\color{blue},
    mathescape,
    showstringspaces=true,
    stepnumber=1,
    xleftmargin=0em,
}

\lstdefinestyle{smallstyle}{
    basicstyle=\footnotesize\ttfamily,
    columns=fullflexible,
    commentstyle=\sffamily\color{black!50!white},
    escapechar=\#,
    % framexleftmargin=1ex,
    % framexrightmargin=1ex,
    keepspaces=true,
    keywordstyle=\color{DodgerBlue4},
    mathescape,
    showstringspaces=true,
    stepnumber=1,
    xleftmargin=0.5em,
    breakindent=0pt,
    % frame=single,
    framextopmargin=10pt,
    framexbottommargin=10pt,
    breaklines=true
}

\lstset{style=smallstyle,language=Haskell}
\newcommand\zlstinline{\let\par\endgraf\lstinline}
\newcommand{\linl}[1]{\lstinline!#1!}
% end listings

\theoremstyle{definition}
\let\example\relax
\let\definition\relax
\newtheorem{example}{Example}[subsection]
\newtheorem{definition}{Definition}[subsection]
\newtheorem{examplep}{Example}[subsection]
\newtheorem{definitionp}{Definition}[subsection]
\crefname{examplep}{example}{examples}  % For lowercase references
\crefname{definitionp}{definition}{definitions}  % For lowercase references
\Crefname{examplep}{Example}{Examples}  % For capitalized references
\Crefname{definitionp}{Definition}{Definitions}  % For capitalized references

\newtheorem*{problemDef}{Problem}
\newcommand{\Oo}{\mathcal{O}}
\newcommand{\Pp}{\mathcal{P}}
\newcommand{\Mm}{\mathcal{M}}
\newcommand{\modeone}{\textbf{mode 1}\,}
\newcommand{\modetwo}{\textbf{mode 2}\,}
\newcommand{\modethree}{\textbf{mode 3}\,}
\newcommand{\modefour}{\textbf{mode 4}\,}
\newcommand{\modefive}{\textbf{mode 5}\,}
\newcommand{\modesix}{\textbf{mode 6}\,}
\newcommand{\f}{\mathsf{f}}
\newcommand{\p}[1]{p(#1)}
\newcommand{\n}[1]{n(#1)}
\newcommand{\arity}{\mathsf{arity}}
\newcommand{\Nat}{\mathbb{N}}
\newcommand{\Int}{\mathbb{Z}}
\newcommand{\nope}{\mathsf{no\ solution}}
\newcommand{\mods}{\mathit{Mods}}
\newcommand{\gram}{\mathit{Gram}}
\newcommand{\metagram}{\mathcal{G}}
\newcommand{\route}{\mathsf{root}}
\newcommand{\prodend}{\mathsf{end}}
\newcommand{\sem}{\mathit{Sem}}
\newcommand{\dom}{\mathit{Domain}}
\newcommand{\production}{\mathit{Prod}}
\newcommand{\lang}{(\mods,\gram,\sem)}
\newcommand{\size}{\mathsf{size}}
\newcommand{\depth}{\mathsf{depth}}
\newcommand{\height}{\mathsf{height}}
\newcommand{\vsep}{\,\,\,|\,\,\,}
\newcommand{\freevars}{\mathit{vars}}
\newcommand{\defn}{\mathsf{def}}
\newcommand{\defns}{\mathsf{definitions}}
\newcommand{\expand}{\mathit{expand}}
\newcommand{\nt}{\mathit{N}}
\newcommand{\lhs}{\mathsf{lhs}}
\newcommand{\rhs}{\mathsf{rhs}}
\newcommand{\term}{\mathit{Term}}
\newcommand{\encode}{\mathit{encode}}
\newcommand{\problemName}{DSL synthesis}
\newcommand{\grp}{\mathit{overlay}}
\newcommand{\solves}{\mathsf{solves}}
\newcommand{\nongen}{\mathsf{nonGeneralizing}}
\newcommand{\bad}{\mathsf{exists\_smaller}}
\newcommand{\fails}{\mathsf{failsToGeneralize}}
\newcommand{\checkBad}{\mathsf{check\_smaller}}
\newcommand{\good}{\mathsf{good\_grammar}}
\newcommand{\checkGood}{\mathsf{check\_good}}
\newcommand{\soln}{\mathsf{exists\_solution}}
\newcommand{\checkSoln}{\mathsf{check\_solution}}
\newcommand{\order}{\mathsf{order}}
\newcommand{\tru}{\mathsf{true}}
\newcommand{\fals}{\mathsf{false}}
\newcommand{\la}{\langle}
\newcommand{\ra}{\rangle}
\newcommand{\res}{\mathbf{reset}}
\newcommand{\guess}{\mathbf{start}}
\newcommand{\start}{\mathbf{start}}
\newcommand{\finish}{\mathbf{finish}}
\newcommand{\interpret}{\mathbf{interpret}}
\newcommand{\guessProd}{\mathbf{prod}}
\newcommand{\checkHit}{\mathbf{hit}}
\newcommand{\checkMiss}{\mathbf{miss}}
\newcommand{\verifySolution}{\mathbf{solve}}
\newcommand{\avoid}{\mathbf{ok}}
\newcommand{\guessRow}{\mathbf{row}}
\newcommand{\down}{\mathbf{down}}
\newcommand{\leftt}{\mathbf{left}}
\newcommand{\rightt}{\mathbf{right}}
\newcommand{\up}{\mathbf{up}}
\newcommand{\stay}{\mathbf{stay}}
\newcommand{\okay}{\mathbf{split}}
\newcommand{\dropFirst}{\mathit{dropFirst}}
\newcommand{\emptyRow}{\mathit{emptyRow}}
\newcommand{\argIdx}{\mathit{idx}}
\newcommand{\topArgs}{\mathit{topArgs}}
\newcommand{\args}{\mathit{args}}
\newcommand{\noArgs}{\mathit{noArgs}}
\newcommand{\argsType}{\mathit{ArgDomain}}
\newcommand{\currDepth}{\mathit{currDepth}}
\newcommand{\offset}{\mathit{offset}}
\newcommand{\fullRow}{\mathit{fullRow}}
\newcommand{\adorn}{\mathit{adorn}}
\newcommand{\dual}{\mathit{dual}}
\newcommand{\mOne}{\mathbf{M1}}
\newcommand{\mTwo}{\mathbf{M2}}
\newcommand{\mThree}{\mathbf{M3}}
\newcommand{\mFour}{\mathbf{M4}}
\newcommand{\mFourA}{\mathbf{M4a}}
\newcommand{\mFourB}{\mathbf{M4b}}
\newcommand{\mFive}{\mathbf{M5}}
\newcommand{\mFiveA}{\mathbf{M5a}}
\newcommand{\mFiveB}{\mathbf{M5b}}
\newcommand{\mSix}{\mathbf{M6}}
\newcommand{\enc}{\mathsf{enc}}
\newcommand{\dec}{\mathsf{dec}}
\newcommand{\extend}{\mathsf{extend}}
\newcommand{\poly}{\text{poly}}
\newcommand{\consistent}{\mathsf{consistent}}
\newcommand{\inconsistent}{\mathsf{inconsistent}}
\newcommand{\rows}{\mathit{prev}}
\newcommand{\Rows}{\mathit{Rows}}

\newcommand{\ignore}[1]{}
\newcommand{\paul}[1]{{\color{blue}paul: #1}}
% \newcommand{\paul}[1]{}

% Algorithms
\algloopdefx{Continue}[0]{\textbf{continue}}

% Clever ref
\crefname{figure}{Figure}{Figures}
\Crefname{figure}{Figure}{Figures}

%% file: abstract.tex
\begin{abstract}
We study the problem of synthesizing domain-specific languages (DSLs)
for few-shot learning in symbolic domains. Given a base language and
instances of few-shot learning problems, where each instance is split
into training and testing samples, the DSL synthesis problem asks for
a grammar over the base language that guarantees that small
expressions solving training samples also solve corresponding testing
samples. We prove that the problem is decidable for a class of
languages whose semantics over fixed structures can be evaluated by
tree automata and when expression size corresponds to parse tree depth
in the grammar, and, furthermore, the grammars solving the problem
correspond to a regular set of trees. We also prove decidability
results for variants of the problem where DSLs are only required to
express solutions for input learning problems and where DSLs are
defined using macro grammars.
\end{abstract}

%% file: introduction.tex
\newcommand{\Cc}{{\mathcal{C}}}
\newcommand{\Hy}{{\mathcal{H}}}
\newcommand{\I}{{\mathcal{I}}}

\section{Introduction}
\label{sec:introduction}
In this work, we are interested in \emph{few-shot learning of symbolic
  expressions}---e.g., learning classifiers expressed in logic which
separate a sample of positive and negative structures or learning
programs that compute a function which is consistent with a few
input-output examples.

When considering a large class of concepts $\Cc$, it is typically
impossible to identify a target concept $c \in \Cc$ from only a small
sample $S$ of instances and hence succeed at few-shot learning, as
there can be many concepts consistent with the sample. In practice,
few-shot symbolic learning, such as program synthesis from examples or
learning invariants for programs, etc., is successful because
engineers identify a much smaller class of concepts $\Hy$— the
\emph{hypothesis class}— and learning proceeds over $\Hy$. Note that
identification of $\Hy$ may not be necessary with a large amount of
data, as in mainstream machine learning. The hypothesis class in
\emph{symbolic} learning is defined using a \emph{language} of
symbolic expressions pertinent to a specific problem domain— such a
language is often referred to as a domain-specific language (DSL)—
which captures \emph{typical concepts} that are \emph{useful} in that
domain. For \emph{few-shot} learning, there is often also an
\emph{ordering} of concepts in $\Hy$, and learning algorithms return—
appealing to Occam's razor— the smallest concept in $\Hy$ with respect
to the order that is consistent with the sample $S$. Typical concept
orderings on symbolic expressions are syntactic length or parse tree
depth.

Program synthesis from examples, especially, is replete with
hand-engineered DSLs for specific applications. For instance, in
spreadsheet automation~\cite{flashfill, flashmeta-prose}— where the
goal is to complete the columns of a spreadsheet in a manner
consistent with a small number of filled-in spreadsheet cells— a
suitable DSL identifies common string-manipulation functions that
occur in realistic spreadsheet programs~\cite{flashfill}. The {\sc
  SyGuS} format (syntax-guided synthesis) for program synthesis makes
explicit the discipline of defining hypothesis classes by using
\emph{grammars} to define DSLs which syntactically constrain the
expressions considered during synthesis~\cite{sygus}, and recent work
explores the case where DSL semantics can be specified in addition to
syntax~\cite{semgus}.

\subsection{Formulation of DSL Synthesis} We consider the problem of
automatically synthesizing DSLs for few-shot symbolic learning. By
solving this problem, few-shot learning in novel domains can be
automated further, with learning algorithms working over hypothesis
classes defined by synthesized DSLs. The first contribution of this
work is a \emph{definition} of the DSL synthesis problem for few-shot
learning. We ask:

\begin{quote}
\begin{center}
  \emph{What formulation of DSL synthesis facilitates few-shot
    learning?}
\end{center}
\end{quote}

We formulate DSL synthesis itself as a \emph{learning} problem. Given
an application domain $D$, the idea is to synthesize a DSL given
\emph{instances of learning problems} from $D$.

Consider a set $\I$ of instances of the few-shot learning problem for
domain $D$. Let us assume a ``base language'' whose syntax is given by
a grammar $G$ over a finite signature providing function, relation,
and constant symbols, and where expressions of $G$ have fixed
semantics. We would like to learn a DSL $\Hy$— with \emph{syntax} and
\emph{semantics} for expressions of $\Hy$ formalized using the base
language $G$— that can be used to effectively solve each of the
problems in $\I$. Each $p\in\I$ is itself a learning problem: it
includes a set of training examples $X_p$ and a set of testing
examples $Y_p$. We require $\Hy$ to solve each problem $p \in \I$ in
the following sense: the \emph{smallest} expressions in $\Hy$,
measured according to a fixed ordering on expressions, that are
consistent with training examples $X_p$ must also be consistent with
testing examples $Y_p$.

This formulation hence asks that a learning bias for $D$ be encoded in
the DSL. Note that a learning algorithm that picks the smallest
consistent expressions in $\Hy$ will in fact solve all the learning
instances from which $\Hy$ was learned. In addition to the
requirements above, the synthesized DSL must satisfy a
``meta-grammar'' constraint $\metagram$— given in the input— which
constrains the kinds of syntax and semantics it may use from the base
language.

A DSL $\Hy$ must satisfy three properties to solve the DSL synthesis
problem:

\begin{itemize}[leftmargin=9pt,itemsep=3pt]
\item[]\textbf{Property 1.} For each instance $p\in \I$, $\Hy$ must
  express some concept $c$ that solves $p$, in the sense that $c$ is
  consistent with both the training and testing examples $X_p$ and
  $Y_p$.
\item[] \textbf{Property 2.} For each instance $p \in \I$, consider
  the smallest expressions $e\in \Hy$, according to the fixed
  ordering, that express concepts $c$ consistent with the training
  examples $X_p$. These concepts $c$ must also be consistent with the
  testing examples $Y_p$.
\item[] \textbf{Property 3.} The definition of $\Hy$ using the base
  language must satisfy constraints $\metagram$.
\end{itemize}

Intuitively, the second property demands that, for any concept $c$
that is expressible in the base language and solves the training set
$X_p$ but does not solve the testing set $Y_p$ for some $p \in P$, the
DSL $\Hy$ must either (1) \emph{not express} $c$ or else (2) ensure
that, if $c$ is expressible, then there is another concept expressible
using a smaller expression— with respect to the fixed expression
ordering— that solves the training and testing sets $X_p$ and
$Y_p$. This property introduces a new learning signal for DSL
synthesis which is not part of existing related problems such as
library learning~\cite{babble,stitch} or the problem addressed by
systems like \textsc{DreamCoder}~\cite{dreamcoder}, where the goal is
to refactor an existing DSL to favor useful concepts. The new signal
is about \emph{inexpressivity}, and asks that some concepts should not
be expressible, or, if they are, they should not be expressible too
succinctly. We wish to find languages specific to domains, and these
may be \emph{less expressive} than the base languages we use to define
them.

\subsection{DSL Synthesis Problems} We study multiple classes of DSL
synthesis problems which are quite general in the sense they are not
tied to particular logics, programming languages, or underlying
theories.

We introduce the four DSL synthesis problems listed below in
increasing order of difficulty. In all problems, there is a fixed base
language $G$, and we are given a set of few-shot learning instances
$\I$ and a meta-grammar constraint $\metagram$. Let us fix an ordering
on expressions.

\begin{itemize}[leftmargin=9pt,itemsep=3pt]
\item[] \textbf{Problem 1:} \emph{Adequate DSL Synthesis.} Is there a
  DSL $\Hy$ satisfying constraints $\metagram$ such that, for every
  instance $p\in \I$, there is at least one concept expressible in
  $\Hy$ that is consistent with the training and testing sets $X_p$
  and $Y_p$? If so, synthesize the DSL.
\item[] \textbf{Problem 2:} \emph{Adequate DSL Synthesis with Macros.}
  The same question above, except posed for DSLs defined using
  grammars with macros.
\item[] \textbf{Problem 3:} \emph{DSL Synthesis.} Is there a DSL $\Hy$
  satisfying constraints $\metagram$ such that, for every instance
  $p\in \I$, there are concepts expressible in $\Hy$ that are
  consistent with the training set $X_p$, and the ones expressible
  most succinctly according to the expression ordering are also
  consistent with the testing set $Y_p$. If so, synthesize the DSL.
\item[] \textbf{Problem 4:} \emph{DSL Synthesis with Macros.} The same
  question above, except posed for DSLs defined using grammars with
  macros.
\end{itemize}

\emph{Adequate DSL Synthesis} captures \textbf{Properties 1} and
\textbf{3} above. It asks whether there is \emph{any} DSL that
expresses a solution for each input learning instance. It is
equivalent to \emph{DSL Synthesis} where the testing sets are empty,
and is therefore independent of the expression ordering. \emph{DSL
  Synthesis} incorporates \textbf{Property 2} and is the problem we
have articulated thus far.

It turns out that standard context-free grammars are not powerful
enough to capture classes of DSLs that use \emph{macros with
  parameters}. Consider a language that allows a macro $f(x_1, x_2)$,
defined by an expression $e$ with free variables $x_1$ and $x_2$,
where $f(e_1, e_2)$ results in the uniform substitution of $e_i$ for
$x_i$ in $e$. Context-free grammars cannot capture such languages when
the terms substituted are arbitrarily large.  We hence also study the
\emph{Adequate DSL Synthesis} and \emph{DSL Synthesis} problems in the
setting where DSLs are defined using more expressive macro
grammars~\cite{fischer-macro-grammar}.

\subsection{Decidability Results} DSL synthesis is a meta-synthesis
problem, i.e., it involves synthesizing a DSL that in turn can solve a
set of few-shot synthesis problems, and it is algorithmically very
complex. It is a natural question whether there are powerful
subclasses where the problem is decidable. More precisely, for a fixed
base language with grammar $G$, using which the hypothesis class $\Hy$
is defined, we would like algorithms that, when given a set of
few-shot learning instances $\I$ and syntax constraint $\metagram$,
either synthesize $\Hy$ that solves the instances and satisfies
$\metagram$ or report that no such DSL exists. We allow the semantics
of DSLs, defined in terms of $G$, to be of \emph{arbitrary length},
which makes decidability nontrivial.

Our second contribution is to show that the four problems identified
above are \emph{constructively decidable} for a large class of base
languages. In particular, we prove that each variant of the DSL
synthesis problem is decidable for a class of base languages recently
shown to have decidable learning problems~\cite{popl22,oopsla23}—
those for which a specialized recursive program can evaluate the
semantics of arbitrarily large expressions using an amount of memory
depending only on a fixed structure over which evaluation occurs.

Our techniques to establish decidability rely on \emph{tree
  automata}--- we show that the class of trees encoding DSLs which
solve the few-shot learning instances is in fact a regular set of
trees. Our automata constructions are significantly more complex than
those for learning expressions~\cite{popl22,oopsla23}, both
conceptually and in terms of time complexity, and we assume these
previous constructions to design tree automata for DSL synthesis.

We point out two factors that make our algorithms complex. First, we
design tree automata which read trees that encode DSLs, and these
automata must verify the existence of arbitrarily large solutions
expressed in these DSLs. Witnessing the non-existence of solutions
involves, in general, examining all expression derivations within an
encoded DSL. Second, we must check the existence of \emph{minimal}
expressions that witness the solvability of each given learning
instance. In particular, we need to use alternating quantification on
trees to capture the fact that there \emph{exists} a solution $e$ for
each instance such that \emph{all} other smaller expressions $e'$ do
not solve the training set. We cannot nondeterministically guess $e$
and $e'$ separately, as they are related. We avoid this guessing of
$e$ and all smaller $e'$ by essentially leveraging a \emph{dynamic
  programming algorithm} for evaluating all expressions derivable in a
given DSL, in the order of parse tree depth. If we execute this
algorithm, then it will check whether the first depth $d$ at which
there exists an expression that solves the training examples is the
same depth at which an expression first solves both the training
\emph{and} testing examples. However, we cannot run the algorithm as
we do not have a DSL. But, it turns out that for \emph{any} DSL of
arbitrary size, the table of results computed by this dynamic
programming algorithm is essentially finite, given a set of learning
instances. The contents of cells in the table come from a finite
domain for any fixed instance, and thus the table rows repeat at a
certain point. Consequently, we can simulate the dynamic programming
algorithm using a tree automaton that computes the table for
arbitrarily large DSLs encoded as trees.

\medskip
\noindent {\bf Contributions.} We make the following contributions:
\begin{itemize}[leftmargin=14pt]
\item A novel formulation of DSL synthesis from few-shot learning
  instances, which asks for a hypothesis class that supports few-shot
  learning in a domain.
\item Decidability results for variants of DSL synthesis over a
  powerful class of base languages. To the best of our knowledge,
  these are the first decidability results for DSL synthesis.
\end{itemize}

The paper is organized as follows. In \Cref{sec:motivation}, we
explore DSL synthesis problems with some examples and applications. In
\Cref{sec:preliminaries}, we establish some definitions and review
background. In \Cref{sec:formulation}, we formulate various aspects of
DSL synthesis. In \Cref{sec:adequate,sec:dsl-synthesis}, we introduce
the \emph{Adequate DSL Synthesis} and \emph{DSL Synthesis} problems
and prove decidability results for a class of base languages whose
semantics can be computed by tree automata and when expression order
is given by parse tree depth. In \Cref{sec:macrogram}, we introduce
variants of the previous problems that use macro grammars and prove
decidability results. \Cref{sec:related} reviews related work and
\Cref{sec:conclusion} concludes. Omitted details throughout the paper
can be found in the appendix.
% Omitted details throughout the paper can be found in the version with
% appendix~\cite{version-with-appendix}.

%% file: motivation.tex
\section{DSL Synthesis: Motivation and Examples}
\label{sec:motivation}

In this section, we motivate the DSL synthesis problem with
applications to few-shot symbolic learning and synthesis, and explain
aspects of our problem formulation.

Computer science is teeming with \emph{symbolic languages} that have
been designed by researchers or engineers, not necessarily to be
highly expressive but, rather, to be well adapted to specific
domains. Such DSLs express \emph{common} properties of the domain
using \emph{succinct} expressions and disallow or make more complex
the representation of concepts that are irrelevant in the domain.

As explained in \Cref{sec:introduction}, we aim to identify DSLs that
facilitate \emph{few-shot learning}. More precisely, we aim to find
DSLs that can succinctly express solutions to typical few-shot
learning problems in a domain while not expressing irrelevant concepts
succinctly. This enables, in particular, \emph{synthesis algorithms}
to solve learning problems by returning the most succinct expressions
in the DSL that satisfy the examples. The design of DSLs which express
the right concepts succinctly is a significant engineering challenge,
and algorithms for learning DSLs from data can thus provide automation
for applying symbolic learning in new domains.

\subsection{DSLs for Program Synthesis}
\label{sec:spreadsheet-programs}

The \textsf{FlashFill} program synthesis
system~\cite{flashfill,flashfillplus} uses a bespoke DSL for
expressing common string transformation programs for Excel
spreadsheets. The original paper describes the careful design of the
DSL and argues that a general-purpose language, like Python, would
make the search too complex by \emph{``allow[ing] the large number of
  functions to be combined in unintuitive ways to produce undesirable
  programs''}. Core features of that DSL include an operation
$\mathsf{SubStr(\mathit{s}, \mathit{P}, \mathit{P})}$ for extracting
substrings from a string $s$ in a spreadsheet cell, where $P$ is a
nonterminal that generates string positions, and an operation
$\mathsf{Pos(\mathit{R},\mathit{R},\mathit{c})}$, which generates a
string position defined as the $c$th one whose substring to the left
matches a regular expression $r_1\in L(R)$ and whose substring to the
right matches a regular expression $r_2\in L(R)$.

We seek automatic design of such DSLs given sample learning
instances. In this case, operations could be defined using
\emph{macros}. For instance, starting from a generic programming
language with recursion, the
$\mathsf{SubStr(\mathit{s}, \mathit{P}, \mathit{P})}$ macro could be
defined using code which recurses over the input $\mathit{s}$ to find
the left position of the substring and return the string ending at the
right position.\footnote{We are motivated by synthesis in domains such
  as spreadsheet programming, but there are other facets to DSL design
  we do not consider. For instance, the FlashFill DSL uses
  ``effectively-invertible'' top operators and
  ``effectively-enumerable'' bottom operators to facilitate
  \emph{faster} search. Such search performance requirements are out
  of scope for this paper.}

The program synthesis literature is replete with DSLs for various
program synthesis tasks--- some examples are the small functional
programming language on lists defined for the {\sc Sketch-n-Sketch}
SVG manipulation
framework~\cite{ProgrammaticAndDirectManipulationTogetherAtLast}, a
DSL designed around an algebra of operators for extracting structured
data from text~\cite{FlashExtract}, a DSL for synthesizing barriers
for crash consistency of file
systems~\cite{SpecifyingAndCheckingFileSystemCrashConsistencyModels},
and a DSL for synthesizing SQL queries from
examples~\cite{SynthesizingHighlyExpressiveSQLQueriesFromInputOutputExamples}.

\subsection{Synthesis Tools Supporting DSLs}
\label{sec:dsltools}

The utility of DSLs for synthesis is reflected in several synthesis
tools and frameworks that provide explicit support for defining DSLs
and their semantics.

The Rosette and Grisette frameworks~\cite{rosette,grisette} support
the definition of solver-aided DSLs, with user-defined syntax and
semantics. The syntax-guided synthesis framework~\cite{sygus} ({\sc
  SyGuS}) supports DSLs with user-specified grammars defining syntax,
with semantics often fixed. The {\sc SemGuS} framework permits
user-specified DSL syntax \emph{and}
semantics~\cite{semgus}. Example-based synthesis in {\sc SyGuS} could
be targeted by our formulation of DSL synthesis with grammars, while
synthesis for DSLs with semantics defined by new functions could be
targeted by our formulation of DSL synthesis with macro grammars.

\subsection{Synthesizing Invariants and Feature Engineering}
\label{sec:feature-engineering}

Symbolic learning often requires a set of base features, e.g., in
learning decision trees or symbolic regression. For instance, base
features might be nonlinear inequalities over numeric variables, with
a symbolic learning algorithm considering Boolean combinations over
the inequalities. In learning \emph{inductive invariants} and
\emph{specifications} of
programs~\cite{ice,angello-contract-synthesis,datadrivenchc},
effective techniques have considered Boolean combinations of
hand-designed base features.

Consider the GPUVerify tool~\cite{gpuverify} which synthesizes
invariants for GPU kernels. It leverages hand-crafted rules for
generating basic candidate invariants and then computes the strongest
conjunction over these using the Houdini
algorithm~\cite{houdini}. Useful hand-crafted rules include, e.g., if
$i \coloneqq i\times 2$ occurs in a loop body, then predicates
expressing that the loop index $i$ is less than a power of $2$ can be
useful, e.g., $i < 2$, $i <4$, $i < 8$, etc., which are common for
tree reduction computations. Or, when threads access fixed-size
contiguous chunks of a shared array, useful predicates express that
the write index is within a bounded region of the array that depends
on the thread identifier, e.g.,
$\mathit{id\times c}\le \mathit{write\_idx}$ and
$\mathit{write\_idx}< \mathit{(id+1)\times c}$, where $\mathit{id}$,
$\mathit{c}$, and $\mathit{write\_idx}$ are, respectively, local
thread identifiers, constant offsets, and indices where writing occurs
in the array. In contrast to designing such predicates by hand, DSL
synthesis would seek to automatically discover them from instances of
invariant learning problems, which could be sampled from a benchmark
of programs to verify. In this setting, we could use a meta-grammar
$\metagram$ to encode the constraint that invariants are
\emph{conjunctions} over a set of base predicates, with a DSL
synthesis algorithm tasked with determining a good set of base
predicates.

DSL synthesis, as proposed in this paper, can serve as a formulation
of the feature synthesis problem, with the goal of discovering
domain-specific symbolic features using few-shot learning instances
drawn from a domain. We can formulate the question as follows: is
there a set of $n$ features, each drawn from a class of functions over
some existing (even more basic) features, such that a fixed class of
symbolic concepts defined over these $n$ features, e.g., specific
Boolean combinations, can solve a given set of few-shot learning
problems sampled from a domain? Such engineered features can then be
used in downstream symbolic learning algorithms for program
synthesis~\cite{sygus} or symbolic
regression~\cite{symbolic-regression}.

\subsection{Library Learning}
\emph{Library learning} is a problem of recent
interest~\cite{babble,stitch,dreamcoder} related to DSL
synthesis. Consider a program synthesizer that solves a class of
problems $\I = \{p_1,p_2, \ldots\}$ in the program synthesis from
examples paradigm. In a library learning phase, we can try to
\emph{refactor} existing solutions to some instances $p_i$ in order to
learn common concepts— a library $L$— that enable compression of the
set of existing solutions to problems from $\I$. \textsc{DreamCoder}
utilizes such refactoring in its \emph{dream phase}, and then, when
synthesizing programs for new learning instances, it uses the library
$L$ to search for solutions~\cite{dreamcoder}. Recent work has used
e-graphs to address library learning with respect to equational
theories~\cite{babble}.

Library learning is similar to the problem of DSL synthesis with
macros introduced in the present work. However, rather than first
synthesizing solutions to instances and then asking whether those
\emph{particular} solutions can be refactored in a synthesized
library, our DSL synthesis problem combines these phases into one---
we ask whether there is a library (realized as a macro grammar) such
that, for each instance, there is \emph{some} solution expressible
using the library. We also introduce the problem of capturing the
domain precisely in a DSL, which may be \emph{less expressive} than
the base language against which it is defined. Furthermore, we prove
decidability results for DSL synthesis problems; previous work on
library learning does not provide such results.

\subsection{Dimensions of the DSL Synthesis Problem}
\label{sec:dimens-dsl-synth}

We collect here the various dimensions of the DSL synthesis problem
formulated in this work and summarize their purposes. We discuss these
further in \Cref{sec:formulation}.

\textbf{1. Expression complexity.} DSLs express domain-specific
concepts succinctly, i.e., by using symbolic expressions which are
syntactically ``simple''. We require a formal way to measure the
complexity of a concept as expressed using an expression in a specific
DSL. Our results measure simplicity using \emph{parse tree depth} with
respect to a DSL, which happens to be relevant in program synthesis,
where a common heuristic search technique is to enumerate the language
of a grammar by increasing depth. Another natural measure would be
parse tree \emph{size}; whether interesting decidability results hold
in this case we leave as an open question.

\textbf{2. Mechanisms for defining DSLs.} As discussed earlier, we
consider two different mechanisms for defining the syntax of DSLs:
regular grammars and macro grammars. Macro grammars are a natural way
to express DSLs with functions that take parameters, and furthermore,
they are more expressive than regular grammars (we give an example
illustrating this in \Cref{ex:macro-more-expressive}).

\textbf{3. Meta grammar.} All variants of the DSL synthesis problems
we introduce have a \emph{meta grammar} in the input. Similar to
grammar constraints in the context of program synthesis, e.g., in
\textsc{SyGuS}~\cite{sygus}, meta-grammar constraints can express
simple kinds of prior knowledge about the space of DSLs one wants to
consider. Some simple regular properties that can be encoded with a
meta grammar include:
\begin{itemize}[leftmargin=14pt]
\item If our DSL has Boolean formulas, we can ensure any mechanism for
  defining them can only define monotonic ones, i.e., those for which
  setting more variables ``true'' cannot flip the formula from
  ``true'' to ``false''. It can do so by disallowing atomic formulas
  that occur under an odd number of negations, for instance.
\item Simple kinds of type checking for a finite number of types,
  e.g., ensuring the generation of Boolean formulas which make
  comparisons on the values of arithmetic expressions.
\item Conjunctive invariants over arbitrary Boolean formulas that
  define ``features'' (\Cref{sec:feature-engineering}).
\end{itemize}
Handling regular syntax constraints makes our results stronger. A user
can omit the meta grammar and a synthesizer simply assumes the most
permissive meta grammar which adds no constraints.

%% file: prelim.tex
\section{Preliminaries}
\label{sec:preliminaries}

\subsection{Alphabets, Tree Automata, and Tree Macro Grammars}
\label{sec:grammars-alphabets}

\begin{definition}[Ranked alphabet]
  A ranked alphabet $\Delta$ is a set of symbols with arities given by
  $\arity : \Delta\rightarrow \Nat$, $\Delta^i\subseteq \Delta$ is the
  subset of symbols with arity $i$, and $x^i$ indicates that symbol
  $x$ has arity $i$. We use $T_\Delta$ to denote the smallest set of
  terms containing $\Delta^0$ and closed under forming new terms using
  symbols of larger arity, e.g.  if $t\in T_\Delta$ and $f\in\Delta^1$
  then $f(t)\in T_\Delta$. For a set of nullary symbols $X$ disjoint
  from $\Delta^0$ we write $T_\Delta(X)$ to mean $T_{\Delta\cup
    X}$. We use \emph{tree} and \emph{term} interchangeably.
\end{definition}

We make use of \emph{tree automata} to design algorithms for DSL
synthesis. We use \emph{two-way alternating tree automata} which
process an input tree by traversing it \emph{up} and \emph{down} while
branching universally in addition to existentially. Transitions are
given by Boolean formulae which describe valid actions the automaton
can take to process its input tree. All automata we deploy have the
form $A=(Q,\Delta,Q^i,\delta)$, with states $Q$, alphabet $\Delta$,
initial states $Q^i\subseteq Q$, and $\delta$ a transition formula,
and they all have acceptance defined in terms of the existence of a
run on an input tree. We refer the reader to~\cite[Chapter 7]{tata}
for background on this presentation of tree automata. We use
$\mathlarger{\mathlarger{\cap}}_iA_i$ and
$\mathlarger{\mathlarger{\cup}}_iA_i$ to denote standard constructions
of automata that accept intersections and unions of languages for
given automata $A_i$. Finally, note that, given a regular tree grammar
$G$ (see~\cite[Chapter 2]{tata}), we can compute in polynomial time a
non-deterministic top-down tree automaton $A(G)$ with $L(A(G))=L(G)$.

\begin{definition}[Tree Macro Grammar]
  A tree macro grammar\footnote{Tree macro grammars are sometimes
    called \emph{context-free tree grammars}, e.g., see~\cite[Chapter
    2.5]{tata}. We prefer \emph{macro grammar} as it evokes
    macros from programming, a useful intuition when using grammars to
    define DSLs.}, or simply \emph{macro grammar}, adds parameters to
  nonterminal symbols of a regular tree grammar. It is a tuple
  $G=(S, N, \Delta, P)$, where: $N$ is a finite set of \emph{ranked}
  nonterminal symbols, $S\in N$ is the starting nonterminal with arity
  $0$, $\Delta$ is a finite ranked alphabet disjoint from $N$, and $P$
  is a finite set of rules drawn from
  $N\times T_{\Delta\cup N}(\Nat)$. We often write rules $(N,t)$ as
  $N\rightarrow t$ and indicate several rules as usual by
  $N\rightarrow t_1 \vsep \cdots\vsep t_k$.
  \label{macro-grammar-defn}
\end{definition}

We refer to nonterminal symbols with arity greater than $0$ as
\emph{macro symbols}. When the nonterminal symbols of a macro grammar
all have arity $0$, i.e. there are no macro symbols, then we recover
the standard concept of a \emph{regular tree grammar} as a special
case.

We consider only \emph{well-formed} macro grammars in the remainder,
i.e. those where right-hand sides of productions refer only to the
parameters for the nonterminal on the left-hand side (if any).

\begin{definition}[Well-formed macro grammar]
  A macro grammar $G=(S,N,\Delta,P)$ is \emph{well formed} if for
  every $(X,t)\in P$ we have that
  $t\in T_{\Delta\cup N}(\{1,\ldots,\arity(X)\})$.
\end{definition}

As usual, we can define the language of a macro grammar to be the set
of ground terms derivable in a finite number of steps by applying
rules starting from $S$— but consider the following
subtlety.

\begin{example}
  \label{ex:outermost}
Consider the macro grammar defined by rules
\[
  S \rightarrow F(H), \qquad F(1)\rightarrow f(1,1), \qquad H\rightarrow a\vsep b,
\]
where integers indicate the parameters for a macro symbol. Observe
that we could apply productions to ``outermost'' macro symbols first,
as in
$S\Longrightarrow F(H) \Longrightarrow f(H, H) \Longrightarrow f(a,
b)$, or we could apply them ``innermost'' first, as in
$S\Longrightarrow F(H) \Longrightarrow F(a) \Longrightarrow f(a,a)$,
or we could mix the two.
\end{example}

\Cref{ex:outermost} illustrates distinct ways to define the language
of a macro grammar, and these choices yield different classes of
languages~\cite{fischer-macro-grammar, tata}. We consider only
outermost derivations and leave an exploration of innermost ones to
future work. \emph{Outermost} derivations are those in which rules are
never applied to rewrite a nonterminal $M$ if it appears as a subterm
of another nonterminal $N$. This can be formalized using
\emph{contexts} (e.g. see~\cite[Chapters 2.1 and 2.5]{tata}) by adding
the requirement that any context used in defining the derivation
relation must not contain nonterminal
symbols.

\begin{definition}[Language of a Tree Macro Grammar]
  The \emph{language} $L(G)\subseteq T_\Delta$ of a macro grammar $G$
  is the set of $\Delta$-terms reachable by applying finitely-many
  productions in an \emph{outermost order} starting from $S$. We often
  write $t\in G$ instead of $t\in L(G)$ to refer to a term in the
  language of $G$. In the remainder, when we say \emph{(macro)
    grammar} we mean \emph{tree (macro) grammar}.
\end{definition}

\subsection{Encoding Grammars as Trees}
\label{sec:encoding-grammars}

\begin{figure}
  \centering
  \hfill
  \begin{minipage}[c]{0.4\linewidth}
    \begin{tabular}[t]{c}
      Macro grammar $G$ \\\\
      $\begin{array}[c]{c c l}
        N_1 &\rightarrow& N_3(N_2) \\
        N_2 &\rightarrow& h(N_1)\vsep a \\
        N_3(1) &\rightarrow& g(1,1)
      \end{array}$
    \end{tabular}
  \end{minipage}%
  \begin{minipage}[c]{0.55\linewidth}
    \begin{tabular}[t]{c}
      \\
      \begin{tikzpicture}[scale=1]
        \node (text) at (3, 0.15) {Encoding $\enc(G)\in T_{\Gamma(\Delta,N)}$} ;
        \node (root) at (0,0.25) {$\route$} ;
        \node (lS) at  (0,-0.5) {$\lhs_{N_1}$} ;
        \node (lF) at (1.25,-0.75) {$\lhs_{N_3}$} ;
        \node (lG1) at (2.5,-1) {$\lhs_{N_2}$} ;
        \node (lG2) at (3.75,-1.25) {$\lhs_{N_2}$} ;
        \node (end) at (5,-1.5) {$\prodend$} ;

        \node (rF) at (-0.5, -1.25) {$\rhs_{N_3}$} ;
        \node (rG) at (-0.5, -2) {$\rhs_{N_2}$} ;

        \node (f) at (0.75, -1.5) {$g$} ;
        \node (1l) at (0.5, -2.25) {$1$} ;
        \node (1r) at (1, -2.25) {$1$} ;

        \node (a) at (2, -1.75) {$h$} ;
        \node (b) at (3.25, -2) {$a$} ;

        \node (rhsN1) at (2, -2.5) {$\rhs_{N_1}$} ;

        \draw (root) -- (lS) ;
        \draw (lS) -- (rF) ;
        \draw (rF) -- (rG) ;
        \draw (lF) -- (f) ;
        \draw (f) -- (1l) ;
        \draw (f) -- (1r) ;
        \draw (lG1) -- (a) ;
        \draw (lG2) -- (b) ;
        \draw (lS) -- (lF) ;
        \draw (lF) -- (lG1) ;
        \draw (lG1) -- (lG2) ;
        \draw (lG2) -- (end) ;

        \draw (a) -- (rhsN1) ;
      \end{tikzpicture}
    \end{tabular}
  \end{minipage}
  \vspace{-0.1in}
  \caption{(Left) Macro grammar $G$ over $\Delta=\{a^0,h^1,g^2\}$ and
    nonterminals $N=\{N_1^0,N_2^0,N_3^1\}$ and (Right) its encoding as
    a tree $\enc(G)\in T_{\Gamma(\Delta,N)}$ over grammar alphabet
    $\Gamma(\Delta, N)$. }
  \label{fig:grammar-tree}
\end{figure}
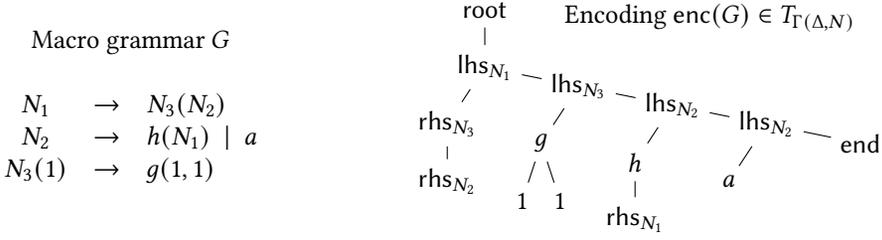

We will define tree automata whose inputs are trees that encode
grammars. \Cref{fig:grammar-tree} shows an example of how we choose to
encode a grammar as a tree; we arrange the grammar rules along the
topmost right-going spine of the tree and use symbols $\lhs_{N_i}$ and
$\rhs_{N_i}$ to distinguish between occurrences of nonterminal $N_i$
on the left-hand and right-hand sides of rules. We use positive
integers to indicate the parameters for macro symbols.  We write
$\Gamma(\Delta,N)$ to denote \emph{grammar alphabets}.

\begin{definition}[Grammar alphabet]
  Given a ranked alphabet $\Delta$ and a set of ranked nonterminal
  symbols $N$ with maximum macro arity $k\in\Nat$, we define its
  \emph{grammar alphabet} as
  \[\Gamma(\Delta, N) \coloneqq
\Delta \sqcup \{\route^1, \prodend^0\}\sqcup \{\lhs_{N_i}^2,
  \rhs_{N_i}^{\arity(N_i)} \,:\, N_i\in N\}\sqcup \{1^0,\ldots,k^0\}.\]
\end{definition}

Any grammar over $\Delta$ using nonterminals $N$ can be encoded as a
term over $\Gamma(\Delta,N)$. We define a mapping $\enc$ from grammars
to the \emph{grammar trees} that encode them, and a mapping $\dec$
from grammar trees back to grammars. These are straightforward and can
be found in \Cref{sec:enc-dec}.
% These are straightforward and can be found in the version with
% appendix~\cite{version-with-appendix}.
We elide the distinction between a \emph{grammar} and its encoding as
a \emph{grammar tree}.

%% file: problem-requirements.tex
\section{Formulating DSL Synthesis}
\label{sec:formulation}

In this section, we introduce a novel formulation of DSL synthesis
that addresses a fundamental aspect of DSLs: namely, a
\emph{domain-specific} language should (a) express relevant domain
concepts and (b) \emph{not} express irrelevant ones, or at least
express them less succinctly than relevant ones. Our formulation is
based on \emph{learning}: expressive power of the DSL must be
carefully tuned to enable solving an input set of few-shot learning
instances. We introduce two distinct mechanisms for specifying which
concepts should be expressed in a DSL, one of which requires certain
concepts to be expressed, addressing (a), and the other, addressing
(b), puts constraints on how succinctly a DSL expresses certain other
concepts, if at all.

This section lays the ground for the formal DSL synthesis problems and
results developed in
\Cref{sec:adequate,sec:dsl-synthesis,sec:macrogram}. We introduce the
two distinct learning signals for our formulation of DSL synthesis and
discuss common aspects of all problems studied in this work.

\subsection{Learning Instances: Expressive Power and Relative Succinctness}
\label{sec:learn-inst-base}

All problems we study involve synthesizing a DSL given \emph{instances
  of few-shot learning problems}.

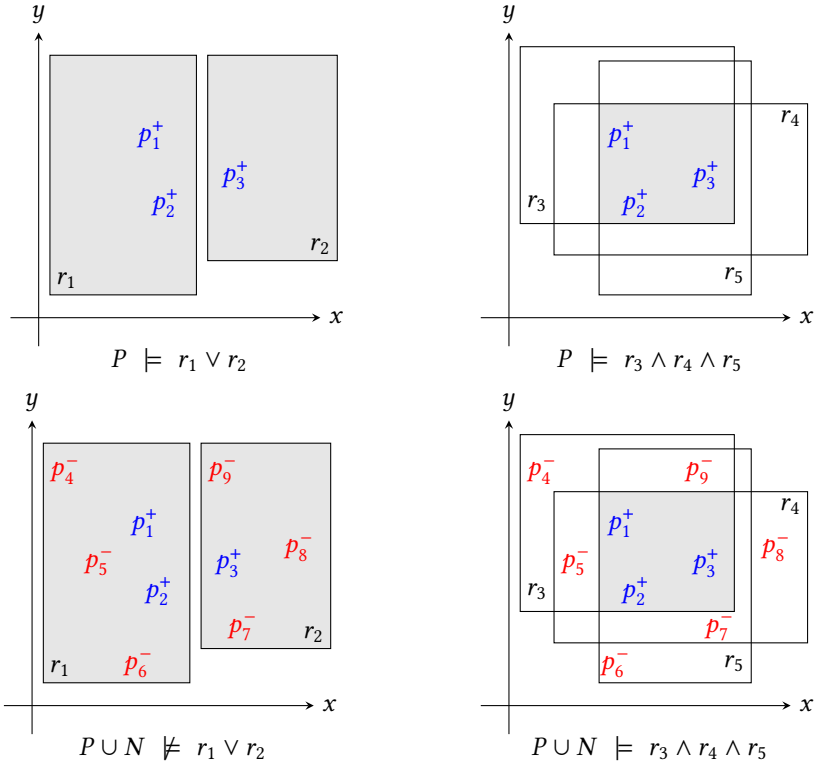
\begin{figure}
  \centering
  \hfill
  \begin{minipage}{0.45\linewidth}
    \begin{tikzpicture}[scale=0.75, axis/.style={thin, ->, >=stealth}]

      % Axes
      \draw[axis] (-0.5, 0) -- (5, 0) node[right] {$x$};
      \draw[axis] (0, -0.5) -- (0, 5) node[above] {$y$};

      \fill[gray!20] (0.2, 0.4) rectangle (2.8, 4.6) ;
      \draw[black] (0.2, 0.4) rectangle (2.8, 4.6) ;
      \fill[gray!20] (3, 1) rectangle (5.3, 4.6) ;
      \draw[black] (3, 1) rectangle (5.3, 4.6) ;

      % Positive points
      \node[blue] at (2, 3.2) {$p_1^{+}$} ;
      \node[blue] at (3.5, 2.5)  {$p_3^+$};
      \node[blue] at (2.25, 2) {$p_2^+$};

      % Labels
      \node[black] at (0.5, 0.65) {$r_1$};
      \node[black] at (5, 1.25) {$r_2$};

      \node at (2.5, -0.75) {$P\,\,\models\,\, r_1\vee r_2$} ;
    \end{tikzpicture}
  \end{minipage}%
  \begin{minipage}{0.45\linewidth}
    \begin{tikzpicture}[scale=0.75, axis/.style={thin, ->, >=stealth}]

      % Axes
      \draw[axis] (-0.5, 0) -- (5, 0) node[right] {$x$};
      \draw[axis] (0, -0.5) -- (0, 5) node[above] {$y$};

      \fill[gray!20] (1.6, 1.65) rectangle (4, 3.75) ;
      \draw[black] (0.2, 1.65) rectangle (4, 4.75) ;
      \draw[black] (1.6, 0.4) rectangle (4.3, 4.5) ;
      \draw[black] (0.8, 1.1) rectangle (5.3, 3.75) ;

      \node[blue] at (2, 3.2) {$p_1^{+}$} ;
      \node[blue] at (3.5, 2.5)  {$p_3^+$};
      \node[blue] at (2.25, 2) {$p_2^+$};

      % Labels
      \node[black] at (0.5, 2) {$r_3$};
      \node[black] at (5, 3.5) {$r_4$};
      \node[black] at (4, 0.75) {$r_5$};

      \node at (2.5, -0.75) {$P\,\,\models\,\, r_3\wedge r_4\wedge r_5$} ;
    \end{tikzpicture}
  \end{minipage}

  \hfill
  \begin{minipage}{0.45\linewidth}
    \begin{tikzpicture}[scale=0.75, axis/.style={thin, ->, >=stealth}]

      % Axes
      \draw[axis] (-0.5, 0) -- (5, 0) node[right] {$x$};
      \draw[axis] (0, -0.5) -- (0, 5) node[above] {$y$};

      \fill[gray!20] (0.2, 0.4) rectangle (2.8, 4.6) ;
      \draw[black] (0.2, 0.4) rectangle (2.8, 4.6) ;
      \fill[gray!20] (3, 1) rectangle (5.3, 4.6) ;
      \draw[black] (3, 1) rectangle (5.3, 4.6) ;

      % Positive points
      \node[blue] at (2, 3.2) {$p_1^{+}$} ;
      \node[blue] at (3.5, 2.5)  {$p_3^+$};
      \node[blue] at (2.25, 2) {$p_2^+$};

      % Negative points
      \node[red] at (0.6, 4.1) {$p_4^-$};
      \node[red] at (3.75, 1.35) {$p_7^-$};
      \node[red] at (3.4, 4.1) {$p_9^-$};
      \node[red] at (1.9, 0.7) {$p_6^-$};
      \node[red] at (1.2, 2.5) {$p_5^-$};

      % Additional points inside/outside the rectangle
      \node[red] at (4.75, 2.75) {$p_8^-$}; % Outside the rectangle

      % Labels
      \node[black] at (0.5, 0.65) {$r_1$};
      \node[black] at (5, 1.25) {$r_2$};

      \node at (2.5, -0.75) {$P\cup N\,\,\not\models\,\, r_1\vee r_2$} ;
    \end{tikzpicture}
  \end{minipage}
  \begin{minipage}{0.45\linewidth}
    \begin{tikzpicture}[scale=0.75, axis/.style={thin, ->, >=stealth}]

      % Axes
      \draw[axis] (-0.5, 0) -- (5, 0) node[right] {$x$};
      \draw[axis] (0, -0.5) -- (0, 5) node[above] {$y$};

      \fill[gray!20] (1.6, 1.65) rectangle (4, 3.75) ;
      \draw[black] (0.2, 1.65) rectangle (4, 4.75) ;
      \draw[black] (1.6, 0.4) rectangle (4.3, 4.5) ;
      \draw[black] (0.8, 1.1) rectangle (5.3, 3.75) ;

      % Positive points
      \node[blue] at (2, 3.2) {$p_1^{+}$} ;
      \node[blue] at (3.5, 2.5)  {$p_3^+$};
      \node[blue] at (2.25, 2) {$p_2^+$};

      % Negative points
      \node[red] at (0.6, 4.1) {$p_4^-$};
      \node[red] at (3.75, 1.35) {$p_7^-$};
      \node[red] at (3.4, 4.1) {$p_9^-$};
      \node[red] at (1.9, 0.7) {$p_6^-$};
      \node[red] at (1.2, 2.5) {$p_5^-$};

      % Additional points inside/outside the rectangle
      \node[red] at (4.75, 2.75) {$p_8^-$}; % Outside the rectangle
le

      % Labels
      \node[black] at (0.5, 2) {$r_3$};
      \node[black] at (5, 3.5) {$r_4$};
      \node[black] at (4, 0.75) {$r_5$};

      \node at (2.5, -0.75) {$P\cup N\,\,\models\,\, r_3\wedge r_4\wedge r_5$} ;
    \end{tikzpicture}
  \end{minipage}
  \caption{A learning instance $I=(P,N)$ where examples are points in
    the plane. Training examples $P$ are the positive points (in blue)
    and testing examples $N$ are the negative points (in red). Top:
    the training examples are satisfied by both $r_1\vee r_2$ and
    $r_3\wedge r_4\wedge r_5$, with $r_i$ signifying membership of a
    point within the depicted rectangle. Bottom: the testing examples
    are satisfied by the conjunction but not the disjunction. }
  \label{fig:learning-instance}
\end{figure}

\begin{definition}[Learning Instance]
  A \emph{learning instance} is a pair $(X,Y)$ consisting of a set $X$
  of \emph{training} examples and a set $Y$ of \emph{testing}
  examples.
\end{definition}

Learning instances, understood as \emph{inputs} to a DSL synthesizer,
contribute \emph{two distinct signals} for construction of a
\emph{domain-specific} language, both of which concern expressive
power of a DSL. The first signal is about \emph{expressivity} and the
second is about \emph{succinctness and inexpressivity} of the DSL. We
discuss these in turn using the learning instance depicted in
\Cref{fig:learning-instance}, which consists of labeled examples of
points in the plane. The training examples $X=P$ are the positive
points and the testing examples $Y=N$ are the negative points. In this
setting, the \emph{domain} we want to capture (using a DSL) is one
where the relevant concepts are sets of points in the plane defined as
intersections of a given set of rectangles. In terms of syntax, such a
domain might correspond to conjunctions of basic predicates, which is
a widespread and useful syntactic bias in various domains, e.g.,
conjunctive invariants in program verification~\cite{houdini}.

\textbf{Expressivity.} The first signal says that, given a learning
instance we would like to solve, an adequate DSL must have enough
power to express \emph{some} concept which solves it\footnote{Similar
  to related problems like library learning, where either specific
  given programs must be expressible or a set of learning problems
  must be solvable using some program (e.g. ~\cite{dreamcoder, stitch,
    babble}).}. That is, we want a DSL which contains some expression
that satisfies all the examples $X\cup Y$. Consider a DSL for the
situation in \Cref{fig:learning-instance} whose language of
expressions consists of Boolean combinations of a fixed set of basic
predicates capturing rectangles in the plane, e.g.
$r_i \coloneqq (1 < x < 3) \wedge (2 < y < 4)$. Any DSL containing an
expression equivalent to $\varphi\coloneqq r_3\wedge r_4\wedge r_5$
(shown in \Cref{fig:learning-instance}) is adequate because $\varphi$
satisfies all training examples $X$ (positive points $P$) and also all
the testing examples $Y$ (negative points $N$). Note that the
expressivity signal treats a learning instance as a set of examples
$X\cup Y$ and thus forgets about the distinction between training and
testing examples.

\textbf{Relative Succinctness and Inexpressivity.} In addition to
expressing relevant concepts, a DSL should also \emph{precisely}
capture those relevant concepts and perhaps little or nothing
else. DSLs need not be fully expressive, and it is in fact a feature
if they avoid expressing irrelevant information which does not reflect
the domain in question\footnote{Imagine a DSL for Excel spreadsheets
  as in~\cite{flashfill}; there are many irrelevant concepts, like
  arithmetic on ASCII codes or reversal of strings, which do not
  reflect typical tasks users of Excel would like to automate. % Or
  % consider a DSL capturing chess strategies which might avoid
  % expressing arbitrary movements on the game board which are
  % illegal.
}. To address this aspect of DSLs, we formulate a \emph{relative
  succinctness} constraint which requires that specific concepts
should be expressed \emph{less succinctly than other concepts (or not
  all)} in a synthesized
DSL.

To understand this inexpressivity signal, let us consider
\Cref{fig:learning-instance} once more and address the purpose of
splitting the instance into training and testing sets $(X,Y)$. Recall
we have assumed a fixed symbolic language of expressions involving
Boolean combinations of some basic rectangles in the plane. This
language is able to represent some regions of the plane using very
succinct expressions and for representing other regions it requires
less succinct expressions. For this example, let us equate
\emph{succinctness} with \emph{syntactic length}, e.g. in
\Cref{fig:learning-instance} the disjunction $r_1\vee r_2$ is
syntactically smaller, and thus more succinct, than the conjunction
$r_3\wedge r_4\wedge r_5$.

Now, imagine we want to solve the learning instance using a fixed
learning algorithm which searches the space of expressible concepts by
first considering small expressions and only later considering large
expressions if no small expression can be found which is consistent
with a small number of examples. Many expression learning algorithms
and synthesis tools use such heuristics, with syntax tree size and
depth being common measures of succinctness. \emph{For such a fixed
  algorithm}, along with a given set of learning instances, our goal
will be to synthesize a DSL such that learning succeeds \emph{using
  the fixed algorithm together with the DSL}, and where \emph{success}
corresponds to finding an expression that is consistent with
\emph{all} the examples in the given learning instance, but where the
learning algorithm itself only considers a smaller number of training
examples. In other words, we ask for a DSL in which the most
\emph{succinct} concepts that it expresses which are consistent with
the training examples $X$ are also consistent with the testing
examples $Y$. Learning algorithms which prefer concepts that are
simple to express, when operating over such DSLs, will discover
solutions which solve training examples \emph{and} which also
generalize to testing examples.

Let us return to the situation depicted in
\Cref{fig:learning-instance}. Suppose that $\psi\coloneqq r_1\vee r_2$
is the shortest expression in our DSL which is consistent with the
(positive) training examples $X$. Since $\psi$ is not consistent with
the (negative) testing examples $Y$, we want to reject this
DSL. Afterall, an algorithm with a bias toward succinct expressions
will select the smaller $\psi$ instead of the larger, but consistent,
$\varphi\coloneqq r_3\wedge r_4\wedge r_5$. The specific learning
instance in \Cref{fig:learning-instance} might favor, for example, a
DSL which allows the expression of conjunctions of base rectangles and
disallows disjunctions entirely. Our formalization of this
inexpressivity signal will in fact also admit DSLs which still express
disjunction, but which make it less succinct to express, rather than
inexpressible. Thus the formal constraint we introduce
(\Cref{generalization}) has to do with the \emph{relative}
succinctness of expressions in a language, with one option for
satisfying the constraint being to not express certain concepts at
all. In the most general case, we are interested in synthesizing DSLs
which meet such relative succinctness constraints for \emph{several
  input learning instances}.

\subsection{Base Language and Consistency}
\label{sec:base-language}

In order to synthesize DSLs, we need some mechanism for defining their
syntax and semantics. For that purpose, we will assume a specific
\emph{base language} with an existing syntax and semantics. We will
define the expressions of our DSLs using the syntax of the base
language, and the semantics of the DSL expressions is, in this way,
inherited from the base language. Furthermore, the base language
specifies what counts as an \emph{example} in a learning instance.

\begin{definition}[Base language]
  The \emph{base language} is specified by a regular tree grammar
  $G=(S,N,\Delta,P)$, a set $\Mm$ of \emph{examples}, and a predicate
  $\consistent \subseteq L(G)\times\Mm$ that holds when an expression
  $e\in G$ is consistent with an example $M\in\Mm$, which we write as
  $\consistent(e,M)$.
  \label{defn:base-language}
\end{definition}

The consistency predicate in \Cref{defn:base-language} abstracts away
details of specific languages while keeping information relevant for
DSL synthesis, namely whether a concept expressed in a given symbolic
language is consistent with some examples from a domain.

We will identify an abstract \emph{concept} $h$ with a \emph{subset}
of examples $\Mm$\footnote{Concepts can be understood independently of
  specific symbolic languages we use to express them. For instance,
  concepts can be arbitrary functions on numbers and DSLs can be
  specific languages expressing programs that compute the
  functions.}. Given a concept $h\subseteq \Mm$, we will say an
expression $e$ is consistent with $h$, written $h\models e$, if
$\consistent(e, M)$ holds for every $M\in h$ and $\consistent(e,M')$
does not hold for any $M'\in \Mm\setminus h$.

\begin{example}[Rectangles in the plane]
  Consider again \Cref{fig:learning-instance} and suppose we have a
  fixed and finite set of basic rectangles $R$. The base language
  might consist of a grammar like the following
  \begin{align*}
    S\rightarrow r\in R \vsep S\vee S \vsep S\wedge S \vsep \neg S,
  \end{align*}
  with examples $\Mm=\mathbb{R}^2\times \{+,-\}$ being labeled points
  in the plane. Whether an example $((x,y),l)$ is consistent with an
  expression $\varphi\in L(S)$ is determined by
  \begin{align*}
    (x,y)\in \llbracket \varphi \rrbracket \quad \text{if}\quad l=+
    \quad \text{and} \quad (x,y)\notin \llbracket \varphi \rrbracket \quad \text{if}\quad l=-,
  \end{align*}
  where
  $\llbracket \cdot\rrbracket : L(S)\rightarrow \Pp(\mathbb{R}^2)$
  interprets each expression as a subset of the plane in the obvious
  way. An abstract concept $h\subseteq \Mm$ in this example
  corresponds to all labeled points consistent with some subset of the
  plane, e.g. for $A\subseteq \mathbb{R}^2$ the concept is
  \[
    h_A=\{((a,b),l) \,:\, (a,b)\in\mathbb{R}^2,\, l=+ \text{ if
    }(a,b)\in A \text{ and otherwise } l=-\}.
  \]
\end{example}

The problems we formulate in
\Cref{sec:adequate,sec:dsl-synthesis,sec:macrogram} are parameterized
by a base language and are thus very general. And our results, as we
will see, hold for a large class of base languages.

With the concept of a base language, we can now formalize both the
mechanisms for defining DSLs and the constraints captured by learning
signals described in \Cref{sec:learn-inst-base}.

\subsection{DSL Spaces and Properties: Adequacy and Generalization}
\label{sec:solut-conc-dsl}

We formalize DSL synthesis problems along two dimensions: (1) the
precise mechanism for specifying DSLs over a base language and (2)
properties required of a DSL.

Along dimension (1), we consider specifying DSLs using either grammars
or (more expressive) macro grammars defined over the base
language. Whether or not macros are allowed, in either case the object
we wish to synthesize is a grammar, which, when combined with the base
language, satisfies a particular solution concept given by properties
(2), which formalize the expressivity and relative
succinctness/inexpressivity signals described in
\Cref{sec:learn-inst-base} as properties of DSLs which we call
\emph{adequacy} and \emph{generalization}.  We start with (1) and then
address (2).

\begin{definition}[DSL space]
  Given a base language with grammar $G'=(S',N',\Delta,P')$, the space
  of DSLs we consider for synthesis is determined by a space of
  grammars $G=(S,N,\Delta,P)$, with $N'\subseteq N$, which define
  productions for new nonterminal symbols $N\setminus N'$. Given a
  synthesized grammar $G$, the resulting DSL is defined by
  $\extend(G,G')\coloneqq (S,N,\Delta, P\cup P')$.
  \label{dsl-space}
\end{definition}
\noindent When the DSL space above is determined by tree grammars $G$ we use the
phrase \emph{DSL synthesis}. When it ranges over tree \emph{macro}
grammars we use the phrase DSL synthesis \emph{with macros}.

\begin{example}[Macro grammars]
  Macro grammars provide a natural way to model classes of expressions
  which take parameters, like functions. Regular grammars do not
  support this, and macro grammars are in fact \emph{more
    expressive}. The macro grammar below defines a set of expressions—
  which is \emph{not} regular— encoding functions over two variables
  $x$ and $y$. These functions sum the results of applying some other
  function to each parameter, and the function must be the \emph{same}
  for each parameter, e.g., things like $h^n(x) + h^n(y)$.
  \begin{align*}
    S\rightarrow F(x,y),\qquad
    F(1,2)\rightarrow F(g(1),g(2)) \vsep F(h(1),h(2)) \vsep 1+2
  \end{align*}
  \label{ex:macro-more-expressive}
\end{example}
\vspace{-0.3in}

Along dimension (2), we consider two properties of DSLs, leading to
weak and strong variants of DSL synthesis. In the weak variant, our
goal is to synthesize a DSL which, for each input learning instance,
expresses some concept that solves the examples it contains. This is
the expressivity signal discussed in \Cref{sec:learn-inst-base}.

\begin{definition}[Expression solution]
  Let $I=(X,Y)$ be a learning instance. We say an expression $e$
  solves $I$ if it is consistent with $X\cup Y$, and $e$ is consistent
  with a set of examples $X$ if we have
  \begin{align*}
    \mathsmaller{\bigwedge}_{M\in X}\,\consistent(e, M).
  \end{align*}
  We write $\solves(e,I)$ and $\solves(e,X)$ when these are true.
\end{definition}

\noindent We call the weak property required of DSLs \emph{adequacy}.

\begin{definition}[Adequacy]
  \label{adequacy}
  Given a set of learning instances $I_1,\ldots, I_n$, a DSL is
  \emph{adequate} if, for each $I_p$, it contains an expression $e$
  such that $\solves(e,I_p)$ holds.
\end{definition}

The second, and stronger, property we introduce for DSLs corresponds
to the relative succinctness of DSL expressions and inexpressivity, as
described in \Cref{sec:learn-inst-base}. This stronger property is
formalized in terms of \emph{concept orderings} induced by DSLs, which
depend on specific ways to measure the complexity (or succinctness) of
expressions.

\begin{definition}[Expression and concept complexity]
  \emph{Expression complexity} for a DSL $G$ is defined by a function
  $c_G : T_\Delta \rightarrow \Nat\cup\{\infty\}$. Note that
  expressions are not themselves ordered by such functions, e.g.,
  parse trees for distinct $e$ and $e'$ may have equal depth. Rather,
  a complexity function partitions expressions and orders cells of the
  partition. Expression complexity induces a notion of complexity on
  concepts $h$ by:
  \[
    c_G(h) = \min\left(\left\{c_G(e) \,:\, e\in T_\Delta,\,\, h\,\models\, e\right\}\right).
  \]
\end{definition}

We single out one particularly common and practically relevant way to
measure expression complexity for a DSL: parse tree
\emph{depth}\footnote{Many synthesis heuristics explore grammars by
  increasing depth.}. This notion is DSL-dependent, in the sense that
an expression $e$ may have a parse of shallow depth in one DSL but
only very large depth in another.

\begin{definition}[Parse tree depth complexity]
  Given a grammar $G$, we define the depth complexity
  $\depth_G : T_\Delta\rightarrow \Nat\cup\{\infty\}$, which for any
  $e\in G$ is the minimum over the set of all parse trees $t$ for $e$
  of the maximum number of nonterminals encountered along any
  root-to-leaf path in $t$\footnote{We omit a formal definition of
    parse trees; see \Cref{sec:enc-dec} for an example.
    % We omit a definition of parse trees; see the version with appendix
    % for an example~\cite{version-with-appendix}.
  }. We let
  $\min(\emptyset)=\infty$. So if there is no parse tree for $e$ in
  $G$ then $\depth_G(e)=\infty$.
\end{definition}

Notions of expression complexity $c_G$, such as parse tree depth,
induce preorders $\le_G$ on expressions of a DSL $G$ and, more
generally, they induce preorders on concepts.

\begin{definition}[Expression and concept orderings]
  By \emph{expression ordering} or \emph{concept ordering}, we mean in
  general a preorder on expressions or concepts. We consider
  expression orderings, written $e \le_G e'$, that are induced by
  expression complexity functions $c_G$ which are relative to a DSL
  $G$, e.g.  $e\le_G e'$ holds if $c_G(e)\le c_G(e')$\footnote{The
    ordering $\le$ on numbers being the usual one extended to
    $\Nat\cup\{\infty\}$.}. Concept orderings are then induced
  similarly by induced concept complexity, i.e. $h\le_G h'$ if
  $c_G(h)\le c_G(h')$.
\end{definition}

We single out the ordering given by parse tree depth complexity, which
appears in our results
(\Cref{sec:adequate,sec:dsl-synthesis,sec:macrogram}).

\begin{definition}[Depth ordering]
  Given a grammar $G$ and an expression $e\in T_{\Delta}$, the
  \emph{depth expression ordering} $e\le_G e'$ holds when
  $\depth_G(e)\le \depth_G(e')$. Similarly, the \emph{depth concept
    ordering} $h\le_G h'$ holds when $\depth_G(h)\le \depth_G(h')$.
  \label{defn:depth-concept-ordering}
\end{definition}

The intuition to take away from this is that DSLs organize a space of
abstract concepts in a way that captures relevant, domain-specific
concepts using symbolic expressions of low complexity.

\begin{example}[Depth concept and expression ordering]
  Again consider \Cref{fig:learning-instance} and suppose we have a
  DSL with an expression grammar $G$ with rules:
  \begin{align*}
    S\rightarrow r\in R \vsep S\wedge S.
  \end{align*}
  There are many inexpressible concepts, e.g. those concepts $h$
  corresponding to non-convex subsets of the plane such as
  $\llbracket r_1\vee r_2\rrbracket$, for which
  $\depth_G(h)=\depth_G(r_1\vee r_2)=\infty$. However, in a more
  permissive DSL $G'$, such as one with the following rules
  \begin{align*}
    S\rightarrow r\in R \vsep S\wedge S \vsep S\vee S \vsep \neg S,
  \end{align*}
  which allows all Boolean combinations of rectangles, we have that
  \[
    \depth_{G'}(r_1\vee r_2)=2 \lneq \depth_{G'}(r_3\wedge r_4\wedge
    r_5)=3.
  \]
  For the DSL $G'$, an algorithm which explores expressions in order
  of increasing depth would find $r_1\vee r_2$ before it finds
  $r_3\wedge r_4\wedge r_5$, though the latter is consistent with the
  testing examples in addition to the training examples.
\end{example}

With our notions of expression ordering and complexity measures like
depth, we can now state a novel property of DSLs, called
\emph{generalization}, which relates to the notions of relative
succinctness and inexpressibility described in
\Cref{sec:learn-inst-base}. The generalization property implies
adequacy (i.e. existence of solutions, see \Cref{adequacy}), and
further requires that some of the most succinct expressions which
solve training examples, where succinctness is measured by a fixed
expression ordering $\le_G$, should also solve testing examples,
i.e. some maximally succinct expression that satisfies all training
examples should generalize to the testing set. Another way of stating
the requirement is that expressions which fail to generalize to
testing examples, though they are consistent with training examples,
should be relatively no more succinct than an expression which does
generalize (or they are not expressed at all).

We call expressions which are consistent with a set of training
examples, but \emph{inconsistent} with a set of testing examples,
\emph{non-generalizing expressions}.

\begin{definition}[Non-generalizing expression]
  Given an instance $I=(X,Y)$, a \emph{non-generalizing expression}
  $e$ is one for which $\solves(e, X)$ holds but $\solves(e,Y)$ does
  not hold.
\end{definition}

\begin{example}
  The expression $r_1\vee r_2$ from \Cref{fig:learning-instance} is
  non-generalizing because it solves the positive examples (training
  set $X$) but does not solve the negative examples (testing set $Y$).
\end{example}

Finally, we can state the generalization property, which lifts the
relative succinctness/inexpressivity requirement to all learning
instances given in the input for a DSL synthesis problem.

\begin{definition}[Generalization]
  \label{generalization}
  Given a set of learning instances $I_1,\ldots,I_n$, a DSL $G$ is
  \emph{generalizing} for an expression ordering $e\le_G e'$ if it is
  adequate (\Cref{adequacy}), and additionally, the following
  holds. For each $I_p$, there is an expression $e\in G$ such that
  $\solves(e,I_p)$ holds, and for all $e'\in G$ which are
  non-generalizing on $I_p$, we have that $e\le_G e'$.
\end{definition}

We now have the primary concepts needed to introduce our DSL synthesis
problems. Before doing so, we discuss (\Cref{sec:constraining}) a
natural mechanism for constraining the space of allowed DSLs, namely,
regular syntax contraints on grammars, and (\Cref{sec:fac}) the scope
of our forthcoming theorems as it pertains to the base language in
terms of which we define DSLs.

\subsection{Constraining the DSL Space}
\label{sec:constraining}

The problems we introduce in
\Cref{sec:adequate,sec:dsl-synthesis,sec:macrogram} are very
general. They are parameterized by a base language, which is used to
define the syntax and semantics of synthesized DSLs
(\Cref{sec:base-language}). In general, the base language should be
relatively expressive, as we do not want to assume knowledge of
relevant concepts in the domain, their relationships, and how much
power is needed to define them. With an expressive base language, e.g.
a Turing-complete programming language, there is no mechanism in the
problem specification (though the algorithms we introduce can easily
output the \emph{syntactically shortest} DSL solving the problem)
which prevents a synthesized DSL from ``memorizing'' solutions to
learning instances by introducing new symbols that exactly define
specific solutions, effectively making the solutions into constants
and therefore as succinct as possible.

Such memorization can be mitigated by enforcing syntactic constraints
on the synthesized DSL. For instance, we might put an upper bound on
the number of rules that any specific nonterminal symbol can use,
ruling out a grammar like
\[ S\rightarrow \mathsf{solution}_1 \vsep \cdots \vsep S\rightarrow
  \mathsf{solution}_k.
\]
Such constraints can be specified in the input, similar to syntax
constraints in program synthesis~\cite{sygus}. In the DSL synthesis
problems we introduce in
\Cref{sec:adequate,sec:dsl-synthesis,sec:macrogram}, syntax
constraints are specified in the input using a grammar, which we refer
to as a \emph{meta-grammar} because it constrains the syntax of
(object) grammars that define DSLs.

\begin{definition}[Meta-grammar]
  By \emph{meta-grammar} $\metagram$ we mean a regular tree grammar
  over a grammar alphabet $\Gamma(\Delta,N)$. We use meta-grammars to
  constrain the syntax of synthesized DSLs.
\end{definition}

Note also that handling a meta-grammar in the input is a
\emph{feature} and makes the problems more general, as it can always
be omitted and an unconstrained grammar can be used by
default\footnote{\Cref{sec:grammar-tree-appendix} shows an example of
  an unconstrained meta grammar.
  % The version with appendix shows an example of an unconstrained meta
  % grammar~\cite{version-with-appendix}.
}.

\subsection{Tree Automaton-Computable Language Semantics}
\label{sec:fac}

Being able to algorithmically check DSLs for adequacy
(\Cref{adequacy}) and generalization (\Cref{generalization}) implies
checking whether a learning instance has \emph{any} solution at
all. In other words, for the problem formulations we pursue, being
able to verify solutions for DSL synthesis implies that symbolic
learning over the underlying base language must be decidable.

Our forthcoming results are \emph{meta-theorems} on the decidability
of DSL synthesis. We obtain as specific instantiations of these
meta-theorems a swath of decidability results on DSL synthesis over a
rich class of base languages recently shown to admit decidable
learning~\cite{popl22,oopsla23}, including finite variable logics,
modal logics, regular expressions, context-free grammars, linear
temporal logic, and some restricted programming languages, among
others. Such languages have semantics which can be computed by a tree
automaton in the sense that the consistency predicate, when
specialized to any \emph{fixed example} $M$, can be computed by a tree
automaton whose size is a function of only the length $|M|$ of an
encoding of the example.

\begin{definition}[Tree Automaton-Computable Semantics]
  Fix a base language consisting of grammar $G=(S,N,\Delta,P)$,
  examples $\Mm$, and predicate
  $\consistent \subseteq L(G)\times \Mm$. We say the language
  semantics \emph{can be evaluated over fixed structures by a tree
    automaton} if, for any example $M\in \Mm$, there are computable
  tree automata $A_M$ and $A_{\neg M}$ that accept, respectively, all
  expressions consistent with $M$ and all expressions inconsistent
  with $M$, i.e., $L(A_M) = \{e\in L(G) \,:\, \consistent(e,M)\}$ and
  $L(A_{\neg M}) = \{e\in L(G) \,:\, \neg \consistent(e,M)\}$. We
  refer to $A_M$ and $A_{\neg M}$ as \emph{example automata}.
  \label{defn:example-automata}
\end{definition}

In the remainder, we introduce various DSL synthesis problems, varying
along the dimensions of (1) plain vs macro grammars and (2) adequacy
and generalization as the synthesis specification, and we prove
decidability results in each setting.

%% file: adequate.tex
\section{Adequate DSL Synthesis}
\label{sec:adequate}

In this section, we introduce the \emph{adequate DSL synthesis
  problem}, the simplest of the problems we consider. Given a set of
few-shot learning instances $I_1,\ldots,I_l$, along with a
meta-grammar constraint $\metagram$, the requirement is to synthesize
an adequate DSL satisfying $\metagram$, i.e. one which contains a
solution for each instance.

\begin{tcolorbox}[
  colback=white,
  colframe=gray,
  boxsep=4pt,
  left=0pt,
  right=0pt,
  top=2pt,
  arc=1pt,
  boxrule=1pt,
  toprule=1pt
  ]
\begin{problemDef}[Adequate DSL Synthesis]
\,\\
  \textbf{Parameters:}
  \begin{itemize}
  \item Finite set of nonterminals $N$
  \item Base language with $G'=(S',N',\Delta,P')$ and $N'\subseteq N$
  \end{itemize}

  \textbf{Input:}
  \begin{itemize}
  \item Instances $I_1, \ldots, I_l$
  \item Meta-grammar $\metagram$ over $\Gamma(\Delta,N)$
  \end{itemize}

  \textbf{Output:} A grammar $G=(S,N,\Delta,P)$ such that:
  \begin{enumerate}
  \item $\extend(G,G')$ is adequate (\Cref{adequacy}) and
  \item $\enc(G)\in L(\metagram)$, i.e. constraints $\metagram$ are satisfied
  \end{enumerate}
\end{problemDef}
\end{tcolorbox}

We now establish decidability of adequate DSL synthesis. First we give
an overview of the proof and then a more detailed construction after.

\subsection{Overview}
\label{sec:decid-gramm-synth}

The proof involves construction of a tree automaton $A$ that reads
finite trees encoding grammars. We design $A$ so that it accepts
precisely those grammars which are solutions to the problem, with
existence and synthesis accomplished using standard algorithms for
emptiness of $L(A)$. The main component in the construction is an
automaton $A_I$ that accepts grammars that, when combined with the
base grammar, are adequate in the sense of \Cref{adequacy}: they
contain an expression consistent with examples $X\cup Y$, where
$I=(X,Y)$. Using these $A_I$, we then construct a final automaton $A$
that is the product over all $A_I$ and also $A(\metagram)$, an
automaton accepting grammars that satisfy the constraint
$\metagram$. This final automaton $A$ accepts exactly the grammars
satisfying $\metagram$ which contain solutions to each instance $I$
and thus which solve the adequate DSL synthesis problem.

As discussed in \Cref{defn:example-automata}, for each example
$M\in X\cup Y$, e.g. a labeled point in the plane, we assume existence
of a (non-deterministic top-down) \emph{example automaton} $A_M$ with
language
\begin{align*}
  L(A_M)&=\left\{e\in L(G')\,:\, \mathsf{consistent}(e,M)\right\},
\end{align*}
i.e., the set of expressions in the base language which are consistent
with the example. If we can in fact construct such automata for a
given base language, then our proof will apply.

The instance automaton $A_I$ reads a grammar tree over alphabet
$\Gamma(\Delta,N)$ and explores potential solutions for $I=(X,Y)$ by
simulating an automaton $A_1$, defined as
\begin{align*}
  A_1\coloneqq \mathop{\mathlarger{\mathlarger{\cap}}}\limits_{M\in X\cup Y} A_M, \quad\text{with
  } L(A_1) = \left\{e\in L(G') \,:\, \solves(e, I)\right\},
\end{align*}
which uses the example automata to accept all expressions in the base
language that solve $I$.

Intuitively, $A_I$ operates by walking up and down the input grammar
tree to nondeterministically guess a parse tree for an expression $e$
that solves $I$. When it reads the right-hand side of a production, it
simulates $A_1$, stopping with acceptance if it completes a parse tree
branch on which $A_1$ satisfies its transition formula. Otherwise it
rejects if $A_1$ is not satisfied, or it continues guessing the
construction of a parse tree if it reads a nonterminal symbol. Each
time it reads a nonterminal, it navigates to the top spine to find
productions corresponding to that nonterminal, and it must guess which
production to use among all those that it finds. If any sequence of
such guesses and simulations of $A_1$ leads to a completed parse tree
which satisfies the transition formulae for $A_1$, then the existence
of a solution in the grammar is guaranteed, and vice versa.

\subsection{Automaton construction} Suppose
$A_{1}=(Q_{1},\Delta,Q^i_{1},\delta_1)$.  We define a two-way
alternating automaton $A_I=(Q, \Gamma(\Delta, N), Q^i, \delta)$. The
automaton operates in two modes. In \textbf{mode 1}, it walks to the
top spine of the input tree in search of productions for a specific
nonterminal. Having found a production, it enters \textbf{mode 2}, in
which it moves down into the term corresponding to the right-hand side
of the production, simulating $A_1$ as it goes.

Below we use $N_i\neq N_j\in N$, $q\in Q_1$, $x, f\in\Delta$, and
$t_1,\ldots,t_r\in T_\Delta(\{\rhs_{N_i} \,:\, N_i\in N\})$. We use an
underscore $"\_"$ to describe a default transition when no other
case matches.

\subsubsection*{\bf Mode 1} Find productions. States are drawn from
 \label{sec:mode-1}
 $\mOne\coloneqq \big(Q_1\times \{\start\}\big) \cup \big(Q_1\times
 N\big)$.

\begin{align*}
  &\delta(\la q, \start \ra, \route) = (\down, \la q, \start\ra)
  &\quad
  &\delta(\la q,N_i\ra, \lhs_{N_j}) = (\up, \la q,N_i\ra)\vee
      (\rightt, \la q, N_i\ra)\\
  &\delta(\la q, \start \ra, \lhs_{N_i}) = (\stay, \la q, N_i\ra) &
  &\delta(\la q,N_i\ra, \_) = (\up, \la q,N_i\ra)\vee
      \left(\vee_{(N_i,\, \alpha)\in P'}(\stay, \la q,
        \alpha\ra)\right)
\end{align*}
$\,\delta(\la q,N_i\ra, \lhs_{N_i}) = (\up, \la q,N_i\ra)\vee (\leftt,
  q)
  \vee (\rightt, \la q,
  N_i\ra)\vee\left(\vee_{(N_i, \,\alpha)\in P'}(\stay, \la
  q, \alpha\ra)\right)$

\subsubsection*{\bf Mode 2} Read productions. States drawn
from \label{sec:bf-mode-2-1} $\mTwo\coloneqq Q_1\cup \big(Q_1\times
\mathit{subterms}(P')\big)$, where
\begin{align*}
  \mathit{subterms}(P') = \mathsmaller{\bigcup}_{(N_i,\, \alpha)\,\in\, P'\,} \mathit{subterms}(\alpha).
\end{align*}

\begin{align*}
  &\delta(q, x) = \delta_1(q,x)
  &\qquad\quad
  &\delta(\la q, \rhs_{N_i}\ra, \_) = (\stay, \la q, N_i\ra)\vee \left(\vee_{(N_i,\, \alpha)\in P'}(\stay, \la q,
    \alpha\ra)\right)\\
  &\delta(q, \rhs_{N_i}) = (\stay, \la q, N_i\ra) &
  &\delta(\la q, f(t_1,\ldots, t_r)\ra, \_) = \adorn(t_1,\ldots,t_r,
  \delta_1(q, f))
\end{align*}

The notation $\adorn(t_1,\ldots, t_r, \varphi)$ represents a
transition formula obtained by replacing each atom of the form
$(i, q)$ in the Boolean formula $\varphi$ by the atom
$(\stay, \la q, t_i\ra)$\footnote{We assume directions in $\delta_1$
  are the numbers $1 (\leftt)$, $2 (\rightt)$, $3\ldots$, etc.}.

Any transition not described by the rules above has transition formula
$\fals$. The full set of states and the initial states for the
automaton are
\begin{align*}
  Q \coloneqq \mOne\cup\mTwo, \qquad Q^i=\{\la q, \start\ra \,:\, q\in Q^i_1\}\subseteq \mOne.
\end{align*}
\begin{lemma}
  $L(A_I) = \big\{t \in T_{\Gamma(\Delta, N)} \,:\,
  \solves(\extend(\dec(t),G'), I)\big\}$.
  \label{lemma:grammar-synthesis}
\end{lemma}
\begin{proof}
  % Follows easily by construction~\cite{version-with-appendix}.
  Follows easily by construction. See
  \Cref{sec:proof-grammar-synthesis}.
\end{proof}

\subsection{Decidability}
\label{sec:decidability}

We use the construction of $A_I$ to prove the following theorem:

\begin{theorem}
  The adequate DSL synthesis problem is decidable for any language
  whose semantics over fixed structures can be evaluated by tree
  automata (\Cref{defn:example-automata}). Furthermore, the set of
  solutions corresponds to a regular set of trees.
\end{theorem}
\begin{proof}
  Given meta-grammar $\metagram$ and instances $I_1, \ldots, I_l$, we
  construct the product
  \begin{align*}
    A\coloneqq A(\metagram)\cap
    \mathsf{convert}\big(\mathop{\mathlarger{\mathlarger{\cap}}}_{p\in
    [l]}A_{I_p}\big)% \footnotemark
    ,
  \end{align*}
  where $\mathsf{convert}(B)$ is a procedure for converting a two-way
  alternating tree automaton $B$ to a top-down non-deterministic
  automaton in time $\exp(|B|)$, as explained in~\cite{vardi-two-way,
    cachat-two-way}. By construction and
  \Cref{lemma:grammar-synthesis} we have
  \[
    L(A)=\big\{t\in L(\metagram) \,:\,
      \mathsmaller{\bigwedge}_{p\in[l]}\solves(\extend(\dec(t), G'), I_p)\big\}.
    \]
    Existence of solutions is decided by an automaton emptiness
    procedure running in time $\poly(|A|)$, and solutions can be
    synthesized by outputting $\dec(t)$ for any $t\in L(A)$ in the
    same time.
\label{adequate-final-steps}
\end{proof}

\begin{corollary}
  Adequate DSL synthesis is decidable in time
  $\emph{\poly}(|\metagram|)\cdot \exp(l\cdot m))$, where $l$ is the
  number of instances and $m$ is the maximum size over all instance
  automata $A_{I}$.
\end{corollary}

\textbf{Remark.} The construction of $A_I$, specifically the
simulation of $A_1$ on a grammar tree, is independent of the learning
problem, and it applies essentially unchanged as a proof of the
following.

\begin{lemma}
  Given a tree automaton $A$, there is a tree automaton $B_A$ that
  accepts an encoding of a grammar $G$ if and only if
  $L(A)\cap L(G)\neq\emptyset$.
\end{lemma}

\begin{proof}
  Follows the same logic as the construction of $A_I$ but leaves out
  handling a base grammar.
\end{proof}
% In the context of this section, the automaton $B_{A}$ corresponds to
% $A_{I}$ and the automaton $A$ corresponds to the automaton $A_1$.

%% file: dsl-synthesis-grammars.tex
\section{DSL Synthesis}
\label{sec:dsl-synthesis}

In this section we introduce the DSL synthesis problem for grammars
and prove decidability for ordering based on expression depth. This
problem asks for a DSL which orders concepts in such a way that
expressions solving learning instances are relatively more succinct
than expressions which fail to generalize on testing sets.

\renewcommand{\thempfootnote}{\fnsymbol{footnote}}
\begin{tcolorbox}[
  colback=white,
  colframe=gray,
  boxsep=7pt,
                  left=0pt,
                  right=0pt,
                  top=3pt,
                  arc=1pt,
                  boxrule=1pt,
                  toprule=1pt
                  ]
\begin{problemDef}[DSL synthesis]
\,\\
  \indent \textbf{Parameters:}
  \begin{itemize}
  \item Finite set of nonterminals $N$
  \item Base language with $G'=(S',N',\Delta,P')$ and $N'\subseteq N$
  \item Expression ordering $\le$% \footnotemark{}
  \end{itemize}

  \textbf{Input:}
  \begin{itemize}
  \item Instances $I_1, \ldots, I_l$
  \item Meta-grammar $\metagram$ over $\Gamma(\Delta,N)$
  \end{itemize}

  \textbf{Output:} A grammar $G=(S,N,\Delta,P)$ such that:
  \begin{enumerate}
  \item $\extend(G,G')$ is adequate (\Cref{adequacy}) and generalizing
    (\Cref{generalization}) and
  \item $\enc(G)\in L(\metagram)$, i.e. constraints $\metagram$ are satisfied
  \end{enumerate}
\end{problemDef}
\end{tcolorbox}

Solutions to DSL synthesis are grammars that make generalizing
expressions appear early in the order and non-generalizing expressions
appear later in the order.

We now prove decidability of DSL synthesis over the class of base
languages described in \Cref{sec:fac} for parse tree depth expression
ordering (\Cref{defn:depth-concept-ordering}). The proof has similar
structure to that of \Cref{sec:adequate}, but requires a new idea to
construct an automaton that can evaluate arbitrarily large grammars
and reason about their induced concept orderings. We introduce the
idea with some intuition about \emph{equivalence of grammars}.

\subsection{Equivalence of Grammars}
\label{sec:compare-grammars}

If there is to exist an automaton that accepts exactly the grammars
solving a \problemName~problem, then it must be possible to partition
the space of grammars into finitely-many equivalence classes based on
their behavior over an instance $I$. For \emph{adequate} DSL synthesis
the ``behavior'' of interest was whether or not a grammar expresses at
least one solution for each learning problem.

What would make two distinct grammars $G_1$ and $G_2$ equivalent with
respect to an instance $I=(X,Y)$ under the stronger requirement of
generalization (\Cref{generalization})? Whether $G_1$ and $G_2$ are
equivalent on $I$ depends on the ease with which they express
different concepts relevant to the examples in $X$ and $Y$. Consider
again the example from \Cref{fig:learning-instance} with $I=(P,N)$.
Suppose $G_1=(S, \{S\}, \Delta, P_1)$ and $G_2=(S,\{S\},\Delta,P_2)$,
with $\Delta=\{\wedge^2,\vee^2,\neg^1\}\cup \{r^0\,:\, r\in R\}$ and
$R=\{r_1,r_2,r_3,r_4,r_5\}$ being a finite set of rectangles in the
plane. Suppose the rules $P_1$ and $P_2$ are
\[
  P_1\,:\,\,S\rightarrow S\wedge S \vsep r\in R\qquad\qquad P_2\,:\,\,
  S\rightarrow \neg(\neg S \vee \neg S) \vsep r\in R.
\]
One way to measure the ease with which $G_1$ or $G_2$ expresses
concepts is to consider expressible concepts indexed by the depths of
the smallest expressions needed to express them. In this case we are
interested in which of the examples $P\cup N$ are included in the sets
defined by Boolean combinations of rectangles in the plane. To each
expression $e$, we can associate a vector $v_e$ of Boolean values, one
per example, which exactly describes how the expression behaves on the
instance $I=(P,N)$. For the expression $r_1$ we have
$v_{r_1}=(1,1,0,1,1,1,0,0,0)$ because points $p_1,p_2$ and
$p_4,p_5,p_6$ fall within $\llbracket r_1\rrbracket$. For any such
``behavioral vector'', we want to know the smallest $d\in\Nat$ for
which it is expressible using an expression of depth $d$ but no
shallower. This depth depends on the grammar. We can encode such
information in grammar-specific tables whose entries are subsets of
behavioral vectors and whose columns and rows are indexed by
nonterminals and increasing integers, respectively, as depicted in
\Cref{fig:recursion-table}.

\begin{figure}[h]
  \centering
  \begin{minipage}[t]{0.2\linewidth}
    \centering
    \begin{tabular}[c]{l}
      $G_1:\,S \rightarrow S\wedge S\vsep r\in R$
    \end{tabular}\vspace{0.1in}
    \begin{tabular}[c]{c|c}
      $d$ & $S$  \\\hline
      1 & $\{v_{r_1},\ldots\}$  \\\hline
      2 & $\{{\color{black}v_{r_1\wedge r_2}},\ldots\}$  \\\hline
      \scalebox{0.7}{$\vdots$} & \scalebox{0.7}{$\vdots$} \\
    \end{tabular}
  \end{minipage}
    % \hfill
  \begin{minipage}[t]{0.05\linewidth}
    \centering\vspace{-0.1in}
    \begin{tabular}[c]{c}
      \\\\\\\\
      $\quad\equiv$
    \end{tabular}
  \end{minipage}
  \hfill
  \begin{minipage}[t]{0.28\linewidth}
    \centering
    \begin{tabular}[c]{l}
      $G_2:\,S \rightarrow \neg(\neg S\vee \neg S)\vsep r\in R$
    \end{tabular}\vspace{0.1in}
    \begin{tabular}[c]{c|c}
      $d$ & $S$  \\\hline
      1 & $\{v_{r_1},\ldots\}$  \\\hline
      2 & $\{{\color{black}v_{\neg(\neg r_1\vee \neg r_2)}},\ldots\}$  \\\hline
      \scalebox{0.7}{$\vdots$} & \scalebox{0.7}{$\vdots$} \\
    \end{tabular}
  \end{minipage}
  \hfill
  \begin{minipage}[t]{0.04\linewidth}
    \centering\vspace{-0.1in}
    \begin{tabular}[c]{c}
      \\\\\\\\
      $\nequiv$
    \end{tabular}
  \end{minipage}
  % \hfill
  \begin{minipage}[t]{0.35\linewidth}
    % \hspace{-0.4in}
    \centering
    \begin{tabular}[c]{l}
      $G_3:\, S \rightarrow S\wedge S\vsep S\vee S\vsep r\in R$
    \end{tabular}\vspace{0.1in}
    \begin{tabular}[c]{c|c}
      $d$ & $S$  \\\hline
      1 & $\{v_{r_1},\ldots\}$  \\\hline
      2 & $\{{\color{red}v_{r_1\vee r_2}},\ldots\}$  \\\hline
      3 & $\{{\color{blue}v_{r_3\wedge (r_4\wedge r_5)}},\ldots\}$  \\\hline
      \scalebox{0.7}{$\vdots$} & \scalebox{0.7}{$\vdots$} \\
    \end{tabular}
  \end{minipage}%
  \caption{Tables for different grammars over the instance $I=(P,N)$
    from \Cref{fig:learning-instance} which capture expressive power
    relative to $I$, stratified by parse tree depth $d$. Tables for
    $G_1$ and $G_2$ are identical. The table for $G_3$ is distinct,
    and registers a non-generalizing expression
    {\color{red}$r_1\vee r_2$} before the first generalizing one
    {\color{blue} $r_3\wedge r_4\wedge r_5$}. }
\label{fig:recursion-table}
\end{figure}

To understand the tables in \Cref{fig:recursion-table}, consider the
simultaneous least fixpoint that defines, for instance, $L(G_1)$ as a
set of $\Delta$-terms. Though it is an infinite set, if we consider
terms modulo equivalence relative to examples in $(P,N)$, then there
are finitely-many equivalence classes— with the $9$ points from
\Cref{fig:learning-instance} there are $2^9$ classes— and the fixpoint
computation needs no more than that number of steps to
terminate. Beyond some depth $k\le 2^9$, expressions of $G_1$ or $G_2$
repeat themselves with regard to $I=(P,N)$. The tables in
\Cref{fig:recursion-table} display classes at their earliest
achievable depth— the entry at row $i$ and column $j$ contains the set
of behavioral vectors achieved first at depth $i$ for nonterminal $j$
(in this case there is only a single nonterminal $S$).

It is easy to see that the tables for $G_1$ and $G_2$ are in fact
identical, since $\varphi\wedge\varphi'$ is logically equivalent to
$\neg (\neg \varphi\vee \neg\varphi')$ for any formulas $\varphi$ and
$\varphi'$. In general, whether or not two syntactically distinct
grammars correspond to identical tables depends on the instance
$I=(X,Y)$, though in the specific case of $G_1$ and $G_2$ they have
identical tables for any $I$ whatsover.

If the ``behavioral vectors'' are drawn from a finite set for any
fixed instance $I$, then we can argue that the rows of these tables
must repeat after a certain finite depth for any learning instance. We
can then ask which among the finitely-many bounded-depth tables over
this domain a given grammar corresponds to, and this gives us a finite
index equivalence relation for grammars. We will design automata which
read grammars (presented as trees) and check which class a grammar
corresponds to by iteratively computing its table row by row. Whether
a grammar satisfies the generalization constraint
(\Cref{generalization}) for a specific instance $I$ can be determined
by checking whether a non-generalizing expression is encountered at an
earlier row in the table than any generalizing one. Consider the table
for $G_3$ in \Cref{fig:recursion-table}, which is distinct from the
tables for $G_1$ and $G_2$. It includes
$v_{r_1\vee r_2}=(1,1,1,1,1,1,1,1,1)$ at depth $2$, which does not
appear in any row of the other tables, and it includes
$v_{r_3\wedge (r_4\wedge r_5)}$, a vector for a generalizing
expression, only at depth $3$. We would want to reject $G_3$ for this
reason.

These tables give us a notion of equivalence that captures whether two
grammars have the same expressive power, parameterized by \emph{parse
  tree depth}, over fixed structures. We use the information captured
by such tables, albeit with a more complex domain $D$, in our
automaton construction for the proof of decidability.

\textbf{Simulating a dynamic program using a tree automaton.}
Computing such a table for a given grammar $G$ can be accomplished
with a \emph{dynamic program} that computes the rows starting from
index $0$ up to a bound which depends on each instance $(X,Y)$. We in
fact use a slight modification of the tables in
\Cref{fig:recursion-table} which take a union of the entries from
previous rows in later rows. We refer to these as \emph{recursion
  tables}. Given $G$, a dynamic program can compute the entry for a
nonterminal $N$ in the recursion table for a given row by taking the
union of all values achievable by $G$ using productions for $N$, with
nonterminals in the right-hand sides of the productions interpreted
using values achieved in earlier rows.

For instance, if the program is running for grammar $G_1$ from
\Cref{fig:recursion-table} and computing the entry of the table $T$ at
row $2$ and column $S$, it computes
\[
  T[2,S]\coloneqq T[1,S]\cup \left\{ v\wedge v'\,:\, v,v'\in T[1, S]\right\}\cup
  \{(\mathds{1}_{\llbracket
    r\rrbracket}(p_1),\ldots,\mathds{1}_{\llbracket r\rrbracket}(p_9)) \,:\,
  r\in R\},
\]
where $\wedge$ above indicates a component-wise conjunction on
vectors.

Notice, however, that for the DSL synthesis problem we do not have a
grammar $G$ over which to run this dynamic program— the grammar is
what we want to synthesize. What we will in fact do is simulate the
dynamic program over a grammar encoded as input to a tree
automaton. In order to simulate this algorithm accurately, the
automaton will \emph{nondeterministically guess} the values for each
row and \emph{verify} the guesses by walking up and down the input
grammar and simulating the \emph{example automata} which can be used
to determine existence of expressions in the grammar which correspond
to specific behavioral vectors for the learning examples..

We now describe the details of this construction further.

\subsection{Automaton Construction}
\label{sec:construction}
The main component of our construction is an automaton $A_I$ accepting
grammars that, when combined with the base grammar, solve an instance
$I=(X,Y)$. Our final automaton $A$ will involve a product over the
instance automata $A_I$.

Similar to the construction in~\Cref{sec:adequate}, we assume
existence of \emph{example automata}
(\Cref{defn:example-automata}). For each example $M\in X\cup Y$, we
assume non-deterministic top-down tree automata $A_M$ and
$A_{\neg{M}}$ over alphabet $\Delta$ with languages
\[
  L(A_M) =\left\{e\in L(G')\,:\, \mathsf{consistent}(e,M)\right\}
  \quad\text{and}\quad L(A_{\neg{M}}) =\left\{e\in L(G')\,:\,
    \neg\mathsf{consistent}(e,M)\right\}.
\]

We can now define \emph{instance automata} $A^I_1$ and $A^I_{2}$:
\begin{align*}
  A^I_1\coloneqq \mathop{\mathlarger{\mathlarger{\cap}}}\limits_{M\in X\cup Y} A_M \qquad\qquad
  A^I_2\coloneqq \left(\mathop{\mathlarger{\mathlarger{\cap}}}\limits_{M\in X} A_M\right)
  \cap \left(\mathop{\mathlarger{\mathlarger{\cup}}}\limits_{M\in Y}
  A_{\neg M}\right).
\end{align*}

We omit the superscript and write $A_1$ or $A_2$ when the instance $I$
is clear. The automaton $A_1$ accepts all generalizing expressions and
the automaton $A_2$ accepts all non-generalizing expressions:
\begin{align*}
  L(A_1) =\left\{e\in L(G')\,:\, \solves(e,I)\right\}
  \qquad\quad  L(A_2) =\left\{e\in L(G')\,:\, \solves(e, X)\wedge\neg\solves(e,Y)\right\}.
\end{align*}

Our goal is to keep track of how these instance automata evaluate over
the expressions admitted by a grammar $G$, in order of increasing
parse tree depths.

Suppose $A_{1}=(Q_{1},\Delta,Q^i_1,\delta_1)$ and
$A_{2}=(Q_{2},\Delta,Q^i_2,\delta_{2})$. Note that, as constructed,
$A_1$ and $A_2$ are non-deterministic top-down automata. We will
consider tables similar to those described in
\Cref{sec:compare-grammars} whose entries range over the powerset
$\mathcal{P}(Q_1\sqcup Q_2)$. On an input grammar tree, our automaton
$A_I$ will iteratively construct the rows of its corresponding
\emph{recursion table} for $I$.

\medskip

\textbf{Recursion Tables.} Let us fix a grammar $G=(S,N,\Delta,P)$. To
define its recursion table $T(G)$, we order its nonterminals as
$N_1,N_2,\ldots,N_{k}$, with $N_1=S$. Now let
$H_i : \mathcal{P}(Q_1\sqcup Q_2)^k\rightarrow \mathcal{P}(Q_1\sqcup
Q_2)$ be the operator defined by the equation
\[
  H_i(R)\, = \bigcup\limits_{(N_i,\,t)\in \,P}\llbracket t\rrbracket^{A_1}_{R}\sqcup
  \llbracket t\rrbracket^{A_2}_{R}, \qquad R\in\mathcal{P}(Q_1\sqcup Q_2)^k.
\]
The notation $\llbracket t\rrbracket^{A_j}_{R}$, for $j\in \{1,2\}$,
denotes the subset of $Q_j$ reachable\footnote{The subset of states
  starting from which the automaton has an accepting run on $t$.} by
running the automaton
\[
  A'_j=\big(Q_j,\Delta\sqcup \big\{\rhs_{N_s} : N_s\in N\big\}, Q^i_j, \delta'_j\big)
\]
on term $t$, where $\delta_j'(q,\rhs_{N_s})=\tru$ for each
$q\in R_s\cap Q_j$ and nonterminal $N_s$ and
$\delta_j'(q,x)=\delta_j(q,x)$ for all other $q\in Q_j, x\in
\Delta$. The intuition is that $H_i$ computes the states of the
instance automata which can be reached by some expression generated by
$N_i$, given an assumption about what states have already been
reached.

The operator
$H : \mathcal{P}(Q_1\sqcup Q_2)^k\rightarrow \mathcal{P}(Q_1\sqcup
Q_2)^k$ defined by $H(R) = (H_1(R),\ldots, H_k(R))$ is monotone with
respect to component-wise inclusion of sets, and thus the following
sequence converges to a fixpoint after $n \le k(|Q_1| + |Q_2|)$ steps:
\[
  (\emptyset,\ldots,\emptyset)\eqqcolon Z_0,\, H(Z_0),\,
  H^2(Z_0),\,\ldots,\, H^n(Z_0)=H^{n+1}(Z_0).
\]

We define the recursion table $T(G)$ as follows. There are $k=|N|$
columns and $n^*+1$ rows, where $n^*\coloneqq k(|Q_1| + |Q_2|)$. The
entry at row $i$, column $j$, denoted $T(G)[i,j]$\footnote{For
  convenience, we index rows starting from zero and columns starting
  from one.}, consists of the subset of values from $Q_1\sqcup Q_2$
that are first achieved at parse tree depth $i$ for nonterminal
$N_j$. For $1\le j \le k$:
\[
  T(G)[0,j] \coloneqq \emptyset \quad \text{and} \quad
  T(G)[i,j] \coloneqq H^i(Z_0)_j\setminus H^{i-1}(Z_0)_j, \quad
  \text{for }\, 0 < i
  \le n^*.
\]
By construction, for the depth concept ordering
(\Cref{defn:depth-concept-ordering}) given by
$\depth_G(e)\le \depth_G(e')$, a grammar $G$ solves $I=(X,Y)$ if and
only if \Cref{eqn:acceptable} holds:
\begin{align}
  \text{\emph{Exists a row $i$ such that}}\,\,F_1\cap
    T(G)[i,1]\neq\emptyset\,\,
  \text{\emph{and for all $0 \le j < i,$}}\,\,F_2\cap
  T(G)[j,1]=\emptyset
  \label{eqn:acceptable}
\end{align}
That is, the grammar $G$ solves $I$ if and only if there is some depth
$i$ at which it generates a solution for $X\cup Y$ and \emph{all}
non-generalizing expressions cannot be generated in depth less than
$i$. Let us say $T(G)$ is \emph{acceptable} if this holds.

We define a tree automaton $A_I$ whose language is
\[
  L(A_I) = \{t \in T_{\Gamma(\Delta, N)} \,:\, \solves(\extend(\dec(t),G'), I)\}.
\]
On input $t\in T_{\Gamma(\Delta,N)}$, the automaton iteratively
guesses the row-by-row construction of the recursion table
$T_t\coloneqq T(\extend(\dec(t),G'))$ starting from row $0$ and
working downward to row $n^*$. At each increasing depth $d$, it keeps
track of which domain values have not yet been achieved and guesses
which new ones can be achieved in depth $d$ using previously computed
values at depths less than $d$. It verifies the guesses by simulating
instance automata $A_i$ on the right-hand sides of grammar rules for
each nonterminal. As it constructs the recursion table, it
simultaneously checks that the table is acceptable according to
\Cref{eqn:acceptable} and accepts or rejects accordingly.

The details of the construction can be found in
\Cref{sec:appendix-section-6}.
% The details of the construction can be found in the version with
% appendix~\cite{version-with-appendix}.
We note that the number of states for $A_I$ is exponential in the
sizes of the instance automata, as the entries of the recursion table
range over subsets of their states.

\subsection{Decidability of DSL Synthesis}
\label{sec:decidability}

\begin{theorem}
  DSL synthesis is decidable with depth concept ordering
  (\Cref{defn:depth-concept-ordering}) for any language whose
  semantics on any fixed structure can be evaluated by tree automata
  (\Cref{defn:example-automata}). Furthermore, the set of solutions
  corresponds to a regular set of trees.\label{main-result}
\end{theorem}
\begin{proof}
  \label{main-proof}
  After construction of $A_I$ the proof is identical to
  \Cref{adequate-final-steps}.
\end{proof}

\begin{corollary}
  For languages covered by \Cref{main-result}, DSL synthesis with
  depth ordering is decidable in time
  $\emph{\poly}(|\metagram|)\cdot\exp(l\cdot\exp(m)))$, where $l$ is
  the number of learning instances and $m$ is the maximum size over
  all instance automata.
  \label{complexity-main-result}
\end{corollary}

To make the content of \Cref{main-result} more explicit, consider an
instance of DSL synthesis which is decidable as a
result. Finite-variable first order logic can be evaluated by tree
automata in the sense of \Cref{defn:example-automata}. In particular,
this means that, given a base language consisting of first-order logic
over, e.g., finite graphs, with formulas working with finitely-many
variables, the DSL synthesis problem with depth ordering is
decidable. The tree automata for evaluating logic formulas have size
exponential in the number of examples $s=|X|+|Y|$ for a learning
instance $(X,Y)$ consisting of Boolean-labeled graphs, exponential in
the size $n$ of graphs $G\in X\cup Y$, and doubly exponential in the
number of variables $k$ that are allowed in formulas. So by
\Cref{complexity-main-result} we have that DSL synthesis with depth
ordering for finite-variable first order logic over finite graphs is
decidable in time $\poly(|\metagram|) \cdot\exp(l\cdot\exp(m))$, where
$m(n,s,k)=\exp(ns^k)$. Similar results follow immediately for several
other languages, e.g., those from~\cite{oopsla23}.

\medskip

\textbf{Remark.} The construction of $A_I$ in \Cref{sec:construction},
specifically the simulation of $A_1^I$ and $A_2^I$ on the grammar
input, is independent of the learning problem and can be used to prove
the following.

\begin{lemma}
  Given tree automata $A$ and $B$, there is a tree automaton $C$ that
  accepts an encoding of a grammar $G$ if and only if there is some
  $i\in\Nat$ such that $L(A)\cap L(G)_i\neq\emptyset$ and
  $L(B)\cap L(G)_i=\emptyset$, where $L(G)_i$ is the set of terms
  obtained at iteration $i$ of the fixpoint computation for $L(G)$.
\end{lemma}

\begin{proof}
  Follows the same logic as the construction of $A_I$ with some
  simplifications.
\end{proof}
% In the context of this section, the automaton $C$ corresponds to
% $A_{I}$, the automaton $A$ corresponds $A_1^I$, and the automaton $B$
% corresponds to $A_2^I$.

\medskip

\textbf{Open Problem. } Finally, we leave open the question of whether
an analogous result to that of \Cref{main-result} holds when the
concept ordering is given by parse tree \emph{size} rather than
depth. It is unlikely that the solution sets would be regular sets, as
this would seem to imply the regularity of sets such as
$\{f(t,t) : t \text{ an arbitrarily large term}\}$, which are not in
fact regular, though the existence of suitable DSLs may still be
decidable.

%% file: macrogram.tex
\section{DSL Synthesis for Macro Grammars}
\label{sec:macrogram}

We now introduce variants of the problems from
\Cref{sec:adequate,sec:dsl-synthesis} which define DSLs using
\emph{macro grammars} (see \Cref{sec:grammars-alphabets} and
\Cref{ex:macro-more-expressive}). % We saw an example of a macro
% grammar in \Cref{ex:macro-more-expressive} and discussed syntax and
% semantics in \Cref{sec:grammars-alphabets}.
We establish decidability results for each variant.  % \vspace{-0.2in}

\subsection{Adequate DSL Synthesis with Macros}
\label{sec:gramm-synth-probl}

\begin{tcolorbox}[
  colback=white,
  colframe=gray,
  boxsep=6pt,
  left=0pt,
  right=0pt,
  top=3pt,
  arc=1pt,
  boxrule=1pt,
  toprule=1pt
  ]
  \begin{problemDef}[Adequate DSL Synthesis with Macros]
\,\\
  \textbf{Parameters:}
  \begin{itemize}
  \item Finite set of nonterminals $N$ containing some macro symbols
  \item Base language $G'=(S',N',\Delta,P')$ with $N'\subseteq N\quad$ // a regular tree grammar
  \end{itemize}

  \textbf{Input:} Instances $I_1, \ldots, I_l$ and meta-grammar
  $\metagram$ over $\Gamma(\Delta, N)$

  \textbf{Output:} Macro grammar $G=(S,N,\Delta,P)$ such that
  \begin{enumerate}
  \item $\extend(G, G')$ is adequate (\Cref{adequacy}) and
  \item $\enc(G)\in L(\metagram)$, i.e. constraints $\metagram$ are satisfied
  \end{enumerate}
\end{problemDef}
\end{tcolorbox}
% \vspace{-0.1in}

\begin{theorem}
  Adequate DSL synthesis with macros is decidable for any language
  whose semantics over fixed structures can be evaluated by tree
  automata (\Cref{defn:example-automata}). Furthermore, the set of
  solutions corresponds to a regular set of trees.
  \label{thm:adequate-macros}
\end{theorem}

Macros lead us to a more complex decision procedure for DSL synthesis—
we briefly explain the adjustments needed to prove decidability— a
complete construction can be found in \Cref{sec:adequate-macros-app}.
% a complete construction can be found in the version with
% appendix~\cite{version-with-appendix}.

The proof of \Cref{thm:adequate-macros} is similar to that of adequate
DSL synthesis from \Cref{sec:adequate}, except $A_I$ uses
exponentially more states to deal with macros. To simulate the
instance automaton $A_1$, it keeps track of \emph{sets} of distinct
expressions generated by a given nonterminal. To see why, consider the
grammar with rules $S\rightarrow H(G),\, H(1)\rightarrow h(1,1),\,$
and $G\rightarrow a\vsep b$. Imagine $A_I$ is reading
$S\rightarrow H(G)$ and checking that $S$ generates an expression
evaluating to $q\in Q_1$, a state of $A_1$. It might check that an
expression generated in $H$ can evaluate to $q$, \emph{assuming} that
the argument $G$ generates an expression evalutating to some
$q_G\in Q_1$. Notice, however, that
$H(G)\Longrightarrow h(G,G)\Longrightarrow h(a,b)$ is a valid
outermost derivation, and $G$ generated two distinct expressions
despite being passed once as an argument to $H$. The automaton can
handle this by tracking \emph{the entire subset of} $Q_1$ that
expressions generated by $G$ can evaluate to. Besides this increase in
states to handle macros, the construction and decision procedure are
similar to that of \Cref{sec:adequate}.

\subsection{DSL Synthesis with Macros}
\label{sec:dsl-synthesis-macros}

Here we define DSL synthesis with macros and state a decidability
result for depth ordering.

\begin{tcolorbox}[
  colback=white,
  colframe=gray,
  boxsep=3pt,
  left=4pt,
  right=4pt,
  top=3pt,
  arc=1pt,
  boxrule=1pt,
  toprule=1pt
  ]
  \begin{problemDef}[DSL Synthesis with Macros]
\,\\
  \textbf{Parameters:}
  \begin{itemize}
  \item Finite set of nonterminals $N$ containing some macro symbols
  \item Base language $G'=(S',N',\Delta,P')$ with $N'\subseteq N\quad$
    // a regular tree grammar
  \item Expression ordering $\le$
  \end{itemize}

  \textbf{Input:} Instances $I_1, \ldots, I_l$ and meta-grammar $\metagram$ over $\Gamma(\Delta,N)$

  \textbf{Output:} Macro grammar $G=(S,N,\Delta,P)$ such that:
  \begin{enumerate}
  \item $\extend(G,G')$ is adequate (\Cref{adequacy}) and generalizing
    (\Cref{generalization}) and
  \item $\enc(G)\in L(\metagram)$, i.e. constraints $\metagram$ are satisfied
  \end{enumerate}
\end{problemDef}
\end{tcolorbox}

The new challenge in this setting is to account for an interaction
between the relative succinctness constraint (\Cref{generalization})
and potentially deep nesting of macro applications. Our result below
holds for classes of macro grammars where all grammar rules have macro
application nesting depths bounded by a constant. Details can be found
in \Cref{sec:section-7-dsl}.
% Details can be found in the version with
% appendix~\cite{version-with-appendix}.
\vspace{-0.05in}

\begin{theorem}
  DSL synthesis with macros is decidable for depth ordering over any
  language whose semantics on fixed structures can be evaluated by
  tree automata (\Cref{defn:example-automata}) for any class of macro
  grammars whose macro nesting depth is bounded. Furthermore, the set
  of solutions corresponds to a regular set of trees.
  \label{thm:order-macro-theorem}
\end{theorem}

Suppose we have a meta grammar $\metagram$. If there is a bound
$b\in \Nat$ on the nesting depth of macros occurring in any grammar
rule of any grammar in $\metagram$, then we can compute $b$ given
$\metagram$ and use it in an automaton construction similar to that of
\Cref{sec:decidability} to synthesize and decide existence of DSLs
with macros which abide by the meta grammar constraint.

Given $b\in \Nat$, we adapt the construction of $A_I$ from
\Cref{sec:decidability}. In order to compute the values achievable by
nonterminals at a given depth in the presence of nested macros, $A_I$
now keeps track of previously computed rows of the table individually,
rather than keeping track of a \emph{union} of previously computed
rows. Because the nesting depth of macros is bounded by $b$, the
automaton need only remember $b$ previous rows in order to accurately
compute the table entries. Additionally, for macro symbols, the
entries of the table correspond to \emph{functions} on sets of values
rather than sets. Besides these differences the construction is
similar to \Cref{sec:decidability}.

\textbf{Open Problem. } Does a result analogous to
\Cref{thm:order-macro-theorem} hold for the case of unbounded macro
nesting depth? If the solution sets are not regular, is the problem at
least decidable?

%% file: related.tex
\section{Related Work}
\label{sec:related}

\emph{Program synthesis and library learning.} Hand-designed DSLs are
crucial in many applications of example-based program synthesis,
e.g.~\cite{ProgrammaticAndDirectManipulationTogetherAtLast,FlashExtract,SpecifyingAndCheckingFileSystemCrashConsistencyModels,SynthesizingHighlyExpressiveSQLQueriesFromInputOutputExamples}. Excel's
FlashFill~\cite{flashfill} enabled program synthesis from few examples
for spreadsheet programming, and a major factor in its success was a
DSL that succinctly captured useful spreadsheet operations. Not only
do DSLs define specific domains, but they can also make synthesis more
tractable by reducing search space sizes. Tradeoffs between expressive
grammars and synthesizer performance were studied for \textsc{SyGuS}
problems in~\cite{padhi-overfitting}.

Work on library learning explores the problem of compressing a given
corpus of programs~\cite{stitch,babble} or refactoring knowledge bases
expressed as logic programs~\cite{knowledge-refactoring}, where the
goal is to find a set of programs which can be composed to generate
the input corpus or which are logically similar or equivalent to the
input, but which also serves to compress it. This contrasts with our
work because DSL synthesis does not assume a given set of solutions
and instead requires a DSL that expresses succinct solutions to given
learning problems. Our formulation also does not require an algorithm
to solve the hard problem of finding semantically \emph{equivalent}
representations of a particular set of solutions. Closer in spirit to
our work are the \textsc{EC$^{\text{2}}$} and \textsc{DreamCoder}
systems~\cite{library-learning-for-neurally-guided-program-induction,dreamcoder},
which learn a library alongside program search to solve classes of
synthesis problems from specific domains. In contrast to our work,
these systems cannot declare there is \emph{no library} over a given
signature which solves a set of learning problems. We also introduce a
new signal related to the \emph{relative succinctness} of DSL
concepts, and our formulation permits expressive power of the base
language to \emph{decrease}.

\emph{Learning grammars and automata.} A large body of work explores
the learning of formal languages, e.g., $L^*$~\cite{lstar} and
RPNI~\cite{rpni} learn regular languages represented by
automata. Recent work studies context-free grammar learning for data
format discovery and fuzz testing~\cite{sen-arvada,
  bastani-input-grammar, miltner2023saggitarius, vstar}. In these
applications, useful grammars express syntactic properties, e.g.,
well-nested XML. This is very different from DSL synthesis, where the
relevant grammar properties depend on language \emph{semantics} and
are determined by given learning instances. Similarly distinct recent
work synthesizes grammars capturing programs that implement specific
exploits~\cite{language-synth-for-exploits}.

Vanlehn and Ball \cite{vanlehn1987} approached grammar learning using
version space algebra~\cite{tom-mitchell-vsa}. The automata in our
proofs represent version spaces: each automaton represents the set of
all grammars which express solutions for some specific learning
instance.

\emph{Applications of tree automata to synthesis.} Tree automata
underlie several classic results on synthesis of finite-state systems,
e.g., the solutions to Church's problem~\cite{church63-journal} by
B{\"u}chi and Landweber~\cite{buchi-landweber-69} and
Rabin~\cite{rabin-69}; such automata read trees encoding system
behaviors, whereas the automata in our work read encodings of DSLs and
the parse trees they
admit. % of DSLs and their expressions in order to
% check consistency with fixed example structures, and we require no
% restriction to efficient classes of structures, in contrast to
% applications of tree automata in model checking over parameterized
% classes of structures~\cite{courcelle}.
Our automata for DSLs rely on existing constructions which, given an
arbitrary finitely-presented structure, produce a tree automaton that
reads parse trees of expressions in the base language and evaluates
them over that structure using memory that is independent of the
expression size but may depend on the structure. These constructions
were used for recent decidability results for learning in
finite-variable logics~\cite{popl22} and other symbolic
languages~\cite{oopsla23}; our results apply to DSL synthesis over all
such languages. Related techniques have been used for decidability
results in program synthesis, e.g., reactive programs from temporal
logic specifications~\cite{csl11-madhu} and uninterpreted programs
from sketches~\cite{cav20}, and as an algorithmic framework in program
synthesis tools,
e.g.,~\cite{isil-relational,wang-blaze-fta,flashmeta-prose,flashfill}.

% \vspace{-0.1in}

%% file: conclusion.tex
\section{Conclusion}
\label{sec:conclusion}

We introduced the problem of synthesizing DSLs from few-shot learning
instances. Our formulation contributes a new \emph{relative
  succinctness} constraint on synthesized DSLs, which requires them to
capture a domain precisely by (a) expressing domain-specific concepts
using succinct expressions and (b) expressing irrelevant concepts
using only less succinct ones, or perhaps not expressing them at
all. DSLs are represented using (macro) grammars which are defined
over a base language that gives semantics to the symbols in the DSL.
% The precise notion of \emph{succinctness} varies as a problem
% parameter.
The DSL synthesis problems we introduce, and the relative succinctness
constraint especially, emphasize the automated construction of DSLs
\emph{for few-shot synthesis}; they ask for DSLs using which specific
synthesis algorithms succeed.

We proved that DSL synthesis is decidable when succinctness
corresponds to small parse tree depth, and the solutions sets
(i.e., DSLs) correspond to regular sets of trees. The result holds for
a rich class of base languages whose semantics over any fixed
structure can be evaluated by a finite tree automaton. We also proved
decidability for variants of the DSL synthesis problem where (a) the
relative succinctness constraint is replaced by a weaker one requiring
only that DSLs express some solution for each instance and (b) where
DSLs are defined using grammars with macros.

Future work should explore practical implementations of DSL synthesis
and probe whether DSLs can indeed be realized by synthesis from
few-shot learning problems and, if so, how much data is needed to
arrive at useful DSLs in specific domains.

%% file: appendix.tex
\section{Section 3}
\label{sec:enc-dec}

\subsection{Encoding grammars as trees}
\label{sec:encoding-grammars-as}

The grammar tree for a grammar $G=(S,N, \Delta, P)$, where we assume
the productions are ordered in a list as
$P=\la P_1, P_2,\ldots, P_s\ra$, is given by
$\enc(G)=\route(\enc_p(P))$, where the spine of productions is
computed recursively on the list $P$ as follows.

\newcolumntype{L}{>{$}l<{$}} % math-mode version of "l" column type
\newcolumntype{C}{>{$}c<{$}} % math-mode version of "l" column type

\begin{center}
\begin{tabular}{L C L}
  \enc_p(\la (N_i,\alpha), L\ra) & = & \lhs_{N_i}(\enc_{t}(\alpha), \enc_{p}(L))) \\
  \enc_p(\la\ra) & = & \prodend \\
  \enc_t(f(t_1,\ldots,t_r)) & = & f(\enc_t(t_1),\ldots,\enc_t(t_r)),
                                  \text{  where } f\in\Delta, \arity(f)=r\\
  \enc_t(N_i(t_1,\ldots,t_r)) & = & \rhs_{N_i}(\enc_t(t_1),\ldots,\enc_t(t_r))\\
  \enc_t(N_i) & = & \rhs_{N_i}\\
  \enc_t(i) & = & i
\end{tabular}
\end{center}

\subsection{Decoding trees into grammars}
\label{sec:decoding-trees-into}

The grammar $G$ corresponding to a grammar tree $t$ of the form
$\route(\lhs_{N_i}(x,y))$ over alphabet $\Gamma(\Delta,N)$, is given
by $\dec(t)=(N_i,N,\Delta,\la (N_i,\dec_t(x)), \dec_p(y)\ra)$, where
$\dec_t$ and $\dec_p$ are computed recursively as follows.
\begin{center}
\begin{tabular}{L C L}
  \dec_p(\lhs_{N_i}(x,y)) & = & \la (N_i, \dec_t(x)), \dec_p(y)\ra\\
  \dec_p(\prodend) & = & \la\ra \\
  \dec_t(f(x_1,\ldots,x_r)) & = & f(\dec_t(x_1),\ldots,\dec_t(x_r)),
                                  \text{   where }f\in\Delta,\arity(f)=r\\
  \dec_t(\rhs_{N_i}(x_1,\ldots,x_r)) & = & N_i(\dec_t(x_1),\ldots,\dec_t(x_r))\\
  \dec_t(\rhs_{N_i}) & = & N_i\\
  \dec_t(i) & = & i
\end{tabular}
\end{center}

Note that when decoding a grammar from a tree we choose to make the
nonterminal for the topmost production in the tree the starting
nonterminal, and when encoding a grammar $(S,N,\Delta,P)$ the first
production in the list of productions $P$ will be the topmost
production.

\section{Section 4}
\label{sec:grammar-tree-appendix}

\subsection{Example meta grammar}
\label{sec:example-meta-grammar}

\Cref{fig:grammar-tree-appendix} shows and example of a fully
permissive meta grammar over alphabet $\Gamma(\Delta,N)$.

\begin{figure}[H]
  \centering
    \begin{tabular}[t]{l}
      Meta grammar $\metagram$ over $\Gamma(\Delta, N)$: \\\\
      $\begin{array}{r c l l}
         S &\rightarrow & \route(\production) & \\
         \production &\rightarrow & \lhs_{N_i}(\term, \production) &
                                                                    N_i\in N \\
           & \vsep & \prodend &
         \\
         \term &\rightarrow & g(\term, \term) & \\
           & | & h(\term) & \\
           & | & N_3(\term) & \\
           & | & \rhs_{N_i} & N_i\in N^{=0} \\
           & | & a & \\
           & | & b &
       \end{array}$
    \end{tabular}\hspace{0.4in}
  \begin{minipage}[t]{0.3\linewidth}
    \begin{tabular}[t]{l}
      Grammar $G$: \\\\
      $\begin{array}[c]{l}
         N_1\, \rightarrow \,\,h(N_2) \\
         N_2\, \rightarrow \,\,h(N_1) \vsep g(a,b)
      \end{array}$\\\\\\
      \end{tabular}
    \begin{tabular}[t]{l}
    Macro grammar $G'$: \\\\
    $\begin{array}[c]{c c l}
      N_1 &\rightarrow& N_3(N_2) \\
      N_2 &\rightarrow& h(N_1)\vsep a \\
      N_3(1) &\rightarrow& g(1,1)
    \end{array}$
  \end{tabular}
\end{minipage}
\begin{tabular}[t]{l}
  \\\\
  \begin{tikzpicture}[scale=1]
      \node (text) at (3, 0.15) {Encoding $\enc(G)\in L(\metagram)$} ;

      \node (root) at (0,0.25) {$\route$} ;
      \node (lN1) at  (0,-0.5) {$\lhs_{N_1}$} ;
      \node (lN21) at (1.25,-0.75) {$\lhs_{N_2}$} ;
      \node (lN22) at (2.5,-1) {$\lhs_{N_2}$} ;
      \node (end) at (3.75,-1.25) {$\prodend$} ;

      \node (hl) at (-0.5, -1.25) {$h$} ;
      \node (rN2) at (-0.5, -2) {$\rhs_{N_2}$} ;

      \node (g) at (0.75, -1.5) {$g$} ;
      \node (al) at (0.5, -2.25) {$a$} ;
      \node (ar) at (1, -2.25) {$b$} ;

      \node (hr) at (2, -1.75) {$h$} ;
      \node (rN1) at (2, -2.5) {$\rhs_{N_1}$} ;

      \draw (root) -- (lN1) ;
      \draw (lN1) -- (hl) ;
      \draw (hl) -- (rN2) ;
      \draw (lN1) -- (lN21) ;
      \draw (lN21) -- (g) ;
      \draw (g) -- (al) ;
      \draw (g) -- (ar) ;
      \draw (lN21) -- (lN22) ;
      \draw (lN22) -- (hr) ;
      \draw (hr) -- (rN1) ;
      \draw (lN22) -- (end) ;
    \end{tikzpicture}
    \end{tabular}
    \begin{tabular}[t]{c}
      \\\\
    \begin{tikzpicture}[scale=1]
      \node (text) at (3, 0.15) {Encoding $\enc(G')\in L(\metagram)$} ;
      \node (root) at (0,0.25) {$\route$} ;
      \node (lS) at  (0,-0.5) {$\lhs_{N_1}$} ;
      \node (lF) at (1.25,-0.75) {$\lhs_{N_3}$} ;
      \node (lG1) at (2.5,-1) {$\lhs_{N_2}$} ;
      \node (lG2) at (3.75,-1.25) {$\lhs_{N_2}$} ;
      \node (end) at (5,-1.5) {$\prodend$} ;

      \node (rF) at (-0.5, -1.25) {$\rhs_{N_3}$} ;
      \node (rG) at (-0.5, -2) {$\rhs_{N_2}$} ;

      \node (f) at (0.75, -1.5) {$g$} ;
      \node (1l) at (0.5, -2.25) {$1$} ;
      \node (1r) at (1, -2.25) {$1$} ;

      \node (a) at (2, -1.75) {$h$} ;
      \node (b) at (3.25, -2) {$a$} ;

      \node (rhsN1) at (2, -2.5) {$\rhs_{N_1}$} ;

      \draw (root) -- (lS) ;
      \draw (lS) -- (rF) ;
      \draw (rF) -- (rG) ;
      \draw (lF) -- (f) ;
      \draw (f) -- (1l) ;
      \draw (f) -- (1r) ;
      \draw (lG1) -- (a) ;
      \draw (lG2) -- (b) ;
      \draw (lS) -- (lF) ;
      \draw (lF) -- (lG1) ;
      \draw (lG1) -- (lG2) ;
      \draw (lG2) -- (end) ;

      \draw (a) -- (rhsN1) ;
    \end{tikzpicture}
  \end{tabular}
  \caption{(Top right) Tree grammars $G$ and $G'$ over
    $\Delta=\{a^0,b^0,h^1,g^2\}$ and nonterminals
    $N=\{N_1^0,N_2^0,N_3^1\}$ and (bottom) their encodings as trees
    $\enc(G)$ and $\enc(G')$ over alphabet $\Gamma(\Delta, N)$. (Top
    left) A meta grammar $\metagram$ over alphabet $\Gamma(\Delta, N)$
    with $\enc(G)\in L(\metagram)$. }
  \label{fig:grammar-tree-appendix}
\end{figure}
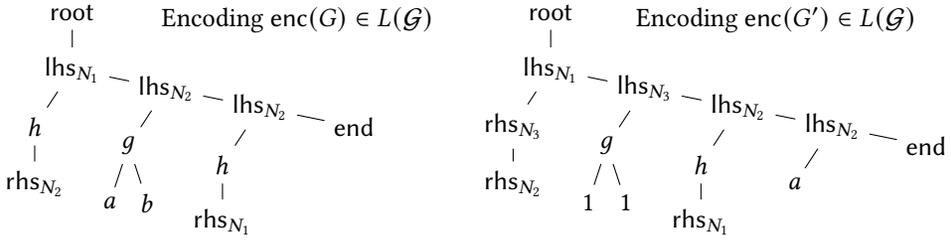

\subsection{Examples of parse trees}
\label{sec:examples-parse-trees}

We omit a (trivial) formalization of \emph{parse trees} for terms and
tree macro grammars. Examples of parse trees for $h(h(h(a)))$ and
$f(f(a))$ in the following macro grammar are shown in
\Cref{fig:macro-parse}. Note that for determining parse tree depth
with respect to macro grammars, we do not consider branches of the
parse tree which record the arguments to a macro (the red edge in
\Cref{fig:macro-parse}). This is because macro definitions may discard
any of their arguments, and in such cases the depth of the parse trees
for the arguments should not count toward the depth of the parse tree
for a macro application because we are considering outermost
derivations.
\begin{align*}
  S &\rightarrow H(A) \vsep f(f(S)) \vsep A \\
  H(1) &\rightarrow h(h(h(1))) \\
  A &\rightarrow a
\end{align*}

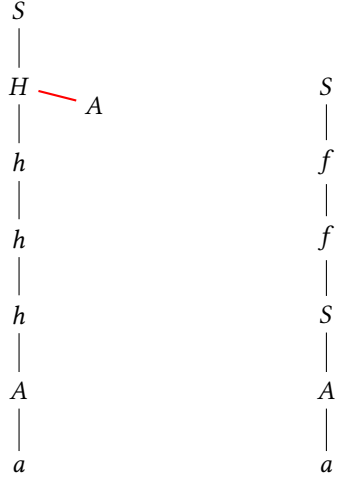
\begin{figure}[t]
  \centering
  \begin{tikzpicture}
    \node (S) at (0,0) {$S$} ;
    \node (M) at (0,-1) {$H$} ;
    \node (MG) at (1,-1.25) {$A$} ;
    \node (h) at (0,-2) {$h$} ;
    \node (h2) at (0,-3) {$h$} ;
    \node (h3) at (0,-4) {$h$} ;
    \node (G) at (0,-5) {$A$} ;
    \node (a) at (0,-6) {$a$} ;

    \draw (S) -- (M) ;
    \draw[red,thick] (M) -- (MG);
    \draw (M) -- (h) ;
    \draw (h) -- (h2) ;
    \draw (h2) -- (h3) ;
    \draw (h3) -- (G) ;
    \draw (G) -- (a) ;
  \end{tikzpicture}\hspace{1in}
  \begin{tikzpicture}
    \node (S) at (0,0) {$S$} ;
    \node (f1) at (0,-1) {$f$} ;
    \node (f2) at (0,-2) {$f$} ;
    \node (S2) at (0,-3) {$S$} ;
    \node (A) at (0,-4) {$A$} ;
    \node (a) at (0,-5) {$a$} ;

    \draw (S) -- (f1) ;
    \draw (f1) -- (f2);
    \draw (f2) -- (S2) ;
    \draw (S2) -- (A) ;
    \draw (A) -- (a) ;
    % \draw (h3) -- (G) ;
    % \draw (G) -- (a) ;
  \end{tikzpicture}
  \caption{Parse trees for terms $h(h(h(a)))$ and $f(f(a))$. Arguments
    for macros are hanging to the right.}
  \label{fig:macro-parse}
\end{figure}

\section{Section 5}
\label{sec:appendix-section-5}

\subsection{Proof of \Cref{lemma:grammar-synthesis}}
\label{sec:proof-grammar-synthesis}

We use the notion of a run for two-way alternating tree automata
from~\cite{cachat-two-way}. Though our presentation in the main text
used terms over ranked alphabets, here it is more convenient to use
the language of labeled trees. Let $W$ be a set of directions using
which our terms over $\Gamma(\Delta,N)$ can be described as finite
$\Gamma(\Delta,N)$-labeled $W$-trees $(T,l)$. Let
$A_X=(Q, \Gamma(\Delta,N), I, \delta)$ and
$A_1=(Q_1,\Delta,I_1,\delta_1)$.

First we argue that
$L(A_X)\subseteq \{t \in T_{\Gamma(\Delta, N)} \,:\,
\solves(\extend(\dec(t),G'), X)\}$. Suppose $t\in L(A_X)$, and let the
corresponding $W$-tree for $t$ be $(T,l)$. Let $(T_r, r)$ be an
accepting run of $A_X$ on $t$, which in our case means that each
branch is finite, and where $T_r$ is a $(W^*\times Q)$-labeled $Z$
tree, for a suitable set $Z$, and $r : T_r \rightarrow (W^*\times Q)$
labels each node of the run with $(p,q)$, whose meaning is that the
automaton passes through input tree node $p$ in state $q$.

The main observation is that there is a subtree of the run $(T_r,r)$
that encodes an expression $e\in L(\dec(t))$ on which the instance
automaton $A_1$ accepts, and thus for which $\solves(e,X)$
holds. There may be multiple such subtrees (which possibly overlap)
and thus there may be multiple such expressions encoded in
$(T_r,r)$. Such an expression can be constructed from the run by
following any branch from a node labeled by $(w, q)$ with
$l(w)=\lhs_{N_i}$ or $l(w)=\route$. For nodes labeled by $(w,q)$ with
$l(w)=x$, for $x\in \Delta\sqcup \{\rhs_{N_i} \,:\, N_i\in N\}$ with
$q\in Q_1$, there is a subset of successors in $T_r$ which can be
followed to construct the subterms for the symbol $x$, possibly
multiple successors in the case of $x\in \Delta^{>0}$, or a single
successor in the case of $x\in \{\rhs_{N_i} \,:\, N_i\in N\}$. For
nodes labeled by $(w,(q_1,\alpha))$ with
$(q_1,\alpha)\in Q_1\times\mathit{subterms}(P')$, a similar procedure
can construct the expression from $\alpha$ by following any branch at
nodes labeled by $\{\rhs_{N_i} \,:\, N_i\in N\}$.

Next we argue that
$L(A_X)\supseteq \{t \in T_{\Gamma(\Delta, N)} \,:\,
\solves(\extend(\dec(t),G'), X)\}$. Suppose $t$ is such that
$\solves(\extend(\dec(t),G'), X)\}$ holds. Let $N_i\in N$ be the
topmost nonterminal in the spine of productions for $t$, which is the
starting nonterminal for $\dec(t)$ and also $\extend(\dec(t),G')$, and
thus there exists $e\in L(N_i)$ within $\extend(\dec(t),G')$ such that
$\solves(e,X)$. Now, it follows that $A_1$ accepts $e$. Since
$e\in L(N_i)$ within $\extend(\dec(t),G')$, the productions used in a
parse $e_p$ for $e$ must appear in $t$ for some production in the
spine, or else they are from $G'$, in which case they appear in the
states of $A_X$. Given the parse $e_p$, we can straightforwardly
construct a finite run $r(e_p)$ for $A_X$ on $t$. This run satisfies
the transitions because $e_p$ uses the productions appearing in $t$ as
well as those in $G'$ and it satisfies the simulation of $A_1$ because
$A_1$ accepts $e$.

\section{Section 6}
\label{sec:appendix-section-6}

We define a tree automaton $A_I$ whose language is
\[
  L(A_I) = \{t \in T_{\Gamma(\Delta, N)} \,:\, \solves(\extend(\dec(t),G'), I)\}.
\]
In the following, we summarize the high-level operation of $A_I$ as it
relates to recursion tables. % ; a detailed construction including states
% and transitions can be found in \Cref{sec:dsl-construction-app}.

\textbf{Overview.} On an input $t\in T_{\Gamma(\Delta,N)}$, the
automaton guesses the construction of the recursion table
$T(t)\coloneqq T(\extend(\dec(t),G'))$ starting from row $0$ and
working downward to row $n^*$. As it guesses the rows, the automaton
checks the newest row can be produced from preceding rows and that
each entry $T(t)[i, j]$ in fact contains
$H^i(Z_0)_j\setminus H^{i-1}(Z_0)_j$, i.e., it is the set of all new
domain values that can be constructed using previously constructed
domain values. To do this, it simulates the instance automata to check
that each value in $T(t)[i,j]$ can be generated using some production
$(N_j,\alpha)$, with each nonterminal that appears in $\alpha$
interpreted as a value in a previously guessed row. Furthermore, to
check that $T(t)[i,j]$ contains \emph{all} new values that can be
generated at stage $i$, the automaton tracks the set of remaining
values that have not yet been generated by stage $i$, namely
$(Q_1\sqcup Q_2) \setminus H^i(Z_0)_j$, and verifies that none of them
are generated in stage $i$.

After constructing each row, the automaton monitors whether column $1$
for the starting nonterminal is \emph{acceptable}, as described
earlier. Recall this corresponds to checking whether a generalizing or
non-generalizing expression is encountered first. If no generalizing
expression has been found yet and a non-generalizing one is found at
row $i$, that is $F_2\cap T(t)[i,1]\neq\emptyset$ and for all
$j \le i$ we have $F_1\cap T(t)[j,1]=\emptyset$, then the automaton
rejects. Otherwise, if a generalizing expression is found at row $i$,
that is $F_1\cap T(t)[i,1]\neq\emptyset$, then the automaton
accepts. Finally, if no generalizing expression is found at row $n^*$,
equivalently $F_1\cap T(t)[i,1]=\emptyset$ for all $0\le i\le n^*$,
then the automaton rejects.

To achieve the functionality described above, the automaton operates
in a few different modes, which we describe below. Recall $k$ is the
number of nonterminals. We use $D$ as a shorthand for
$\mathcal{P}(Q_1\sqcup Q_2)$, and so rows of the recursion table are
drawn from $D^k$, and we abuse notation by writing $L\cup L'$ for
component-wise union over vectors $L,L'\in D^k$. Besides control
information for making transitions between the modes described below,
the automaton maintains $3$ vectors $L,C,R\in D^k$, where $L$ tracks
the component-wise union over all previously constructed rows, $C$
tracks the current row, and $R$ tracks the remaining values which have
not yet been obtained.

In \textbf{mode 1}, $A_I$ moves to the top of the input tree $t$. From
\textbf{mode 1} it enters \textbf{mode 2}, in which it guesses which
element $C\in D^k$ appears as the next row of $T(t)$, with the
requirement that it contain at least one non-empty component, i.e.
$C\neq \{\emptyset\}^k$. This non-emptiness requirement is what makes
the automaton reject if it constructs the entire table and has not yet
accepted. In \textbf{mode 3} it traverses the right spine of the tree
to verify the guess $C$. For each nonterminal $N_i$ encountered along
the right spine of $t$, this involves guessing which subset
$U\subseteq C_i$ a given production for $N_i$ should reach, given a
vector $L\in D^k$ consisting of all previously reached values for all
nonterminals. In \textbf{mode 4} and \textbf{mode 5} it attempts to
verify these guesses. In \textbf{mode 4}, it simulates the modified
instance automaton $A'_1$ on the right-hand side of a given production
to check that all values in $U$ are reachable, assuming those in $L$
are reachable. In \textbf{mode 5}, it simulates $A'_2$ to check that
all values in $R_i\coloneqq (Q_1\sqcup Q_2) \setminus (L\cup C)_i$ are
not reachable, again assuming those in $L$ are reachable. Note that
after the automaton reaches the end of the productions on the right
spine of $t$, it simulates \textbf{modes 4} and \textbf{5} as if the
productions $P'$ from the base grammar $G'$ were present in the
tree. Finally, it enters \textbf{mode 6} to check if the partially
guessed column corresponding to the starting nonterminal is already
acceptable, and if so it accepts. Otherwise it verifies that no
non-generalizing expression has yet been constructed and enters
\textbf{mode 1} to return to the root of $t$.

We note that the number of states for $A_I$ is exponential in the
sizes of the instance automata, as the entries of the recursion
table range over subsets of their states.

\subsection{Construction of $A_I$}
\label{sec:dsl-construction-app}

We now precisely define the two-way alternating tree automaton
$A_I=(Q, \Gamma(\Delta, N), Q_i, \delta)$ with acceptance defined by
the existence of a finite run satisfying the transition formulas. It
follows the logic described in the overview from the main text.

We describe the states $Q$ and their transition formulas grouped by
functionality. We assume $k$ nonterminal symbols. Below we use
$i,j\in[k]$, $m\in\{1,2\}$, $u\in Q_1\sqcup Q_2$, $u_1\in Q_1$,
$u_2\in Q_2$, $U,V\in D$, $L, R, R', C, C', W\in D^k$,
$N_i, N_j\in N$, $f\in \Delta^r$, and
$t_1,\ldots,t_r\in T_\Delta(\{\rhs_{N_i} \,:\, N_i\in N\})$. We use an
underscore $"\_"$ to describe a default transition when no other case
matches.

\subsubsection*{\bf Mode 1} Reset to the top of the input tree.
\label{sec:mode-1} States are drawn from:
\begin{align*}
\mOne\coloneqq (D^k)^3\times \{\res, \guess\}.
\end{align*}

\begin{itemize}[leftmargin=12pt]
\setlength\itemsep{0.5em}
\item $\delta(\la L,C,R, \res\ra, \route) = (\down, \la L,C,R, \guess\ra)$
\item $\delta(\la L,C,R, \res\ra, \_) = (\up, \la L,C,R, \res\ra)$
\item $\delta(\la L,C,R,\guess\ra, \lhs_{N_i}) = (\stay, \la L,C,R,i,\guessRow\ra)$
\end{itemize}

\subsubsection*{\bf Mode 2} Guess next row of the recursion
table.\label{sec:bf-mode-3} States drawn from:
\begin{align*}
\mTwo\coloneqq (D^k)^3\times[k]\times \{\guessRow\}.
\end{align*}

\begin{itemize}[leftmargin=12pt]
\item $\delta(\la L,C,R,i,\guessRow\ra, \_) = \bigvee_{(C', R')\in\, \okay(R)}
    (\stay, \la L\cup C, C', C', R', i,\guessProd\ra)$
  \item[] with $\okay(R) \coloneqq \big\{(C',R') \in D^k\times D^k
      \,:\,
      C'\cup R'=R,\, C'\neq \{\emptyset\}^k\big\}$
\end{itemize}

\subsubsection*{\bf Mode 3} Guess the contributions of productions to
each row entry. States drawn from: \label{sec:bf-mode-3.1}
\begin{align*}
\mThree\coloneqq (D^k)^4\times[k]\times \{\guessProd\}.
\end{align*}

\begin{itemize}[leftmargin=12pt]
\setlength\itemsep{0.5em}
\item
  $\delta(\la L,C,W,R,i, \guessProd\ra, \lhs_{N_j}) =
   \bigvee_{\left\{(U, V) \,:\, U\cup V = C_j\right\}} (\leftt, \la
  L,U, \checkHit\ra) \wedge (\leftt, \la L,R_j, \checkMiss\ra) \wedge\\
  (\rightt, \la L, \la C_1,\ldots,C_{j-1},V,\ldots,C_k\ra, W, R,i, \guessProd\ra)$
\item $\delta(\la L, C, W, R,i, \guessProd\ra, \prodend) =\\
  \text{ if }\,\,\exists (N_j\in N\setminus N').\, C_j\neq \emptyset\\
  \text{ then}\,\,\fals \\ \text{ else}$
\item[]\hspace{0.2in}
  \begin{tabular}[t]{l}
    $\Big((\stay, \la L_i\cup W_i, \verifySolution\ra)$
    $\vee \big((\stay, \la L, W, R, \res\ra)\wedge  (\stay, \la L_i\cup W_i,
  \avoid\ra)\big)\Big)$ \\
  $\wedge\Big(\bigwedge_{N_j\in
    N'}\bigwedge_{u\in C_j}\bigvee_{(N_j,\alpha)\in P'}(\stay, \la L, \{u\},
  \alpha,\checkHit\ra)\Big)$ \\
  $\wedge\Big(\bigwedge_{N_j\in N'}\bigwedge_{(N_j,\alpha)\in P'}(\stay, \la L, R_j,
  \alpha,\checkMiss\ra)\Big)$
  \end{tabular}
\end{itemize}

\subsubsection*{\bf Mode 4} Check a set of values can be
reached. \label{sec:bf-mode-3.2} States drawn from:

\begin{align*}
  \mFour &\coloneqq \mFourA\cup \mFourB \\
  \mFourA &\coloneqq D^k\times D\times \{\checkHit\} \cup ((Q_1\sqcup
  Q_2)\times (D^k \times \{1,2\})) \\
  \mFourB &\coloneqq \big(D^k\times D\times
            \mathit{subterms}(P')\times \{\checkHit\}\big) %\\ &\qquad
                                                                \cup \big((Q_1\sqcup
  Q_2)\times (D^k \times \{1,2\}\times \mathit{subterms}(P'))\big),\\
  &\text{where }\mathit{subterms}(P') = \mathsmaller{\bigcup}_{(N_i,\,\alpha)\,\in\, P'\,} \mathit{subterms}(\alpha)
\end{align*}
Transitions for $\mFourA$:
\begin{itemize}[leftmargin=12pt]
\setlength\itemsep{0.5em}
\item $\delta(\la L,U, \checkHit\ra, \_) = \bigwedge_{u_1\in U\cap Q_1}
  (\stay, \la u_1, \la L, 1\ra \ra) \wedge \bigwedge_{u_2\in U\cap Q_2}
  (\stay, \la u_2, \la L, 2\ra \ra)$
\item $\delta(\la u, \la L, m\ra\ra, x) = \adorn(\la L, m\ra,
  \delta_{m}(u, x)),\quad x\in\Delta$
\item $\delta(\la u, \la L, m\ra\ra, \rhs_{N_i}) = \tru \quad\, \text{if
  } u\in L_i$
\item $\delta(\la u, \la L, m\ra\ra, \rhs_{N_i}) = \fals \quad \text{if
  } u\notin L_i$
\end{itemize}
Transitions for $\mFourB$:
\begin{itemize}[leftmargin=12pt]
\setlength\itemsep{0.5em}
\item $\delta(\la L,U,\alpha, \checkHit\ra, \_) = \bigwedge_{u_1\in U\cap Q_1}
  (\stay, \la u_1, \la L, 1, \alpha\ra \ra) \wedge \bigwedge_{u_2\in U\cap Q_2}
  (\stay, \la u_2, \la L, 2, \alpha\ra \ra)$
\item $\delta(\la u, \la L, m, f(t_1,\ldots,t_r)\ra\ra, \_) =
  \adorn'(\la L, m, t_1,\ldots, t_r\ra,
  \delta_{m}(u, f))$
\item $\delta(\la u, \la L, m, \rhs_{N_i}\ra\ra, \_) = \tru \quad\, \text{if
  } u\in L_i$
\item $\delta(\la u, \la L, m, \rhs_{N_i}\ra\ra, \_) = \fals \quad \text{if
  } u\notin L_i$
\end{itemize}

The notation $\adorn(s, \varphi)$ represents the transition formula
obtained by replacing each atom $(i, q)$ in the Boolean formula
$\varphi$ by the atom $(i, \la q, s\ra)$. The notation
$\adorn'(\la s, t_1,\ldots, t_r\ra, \varphi)$ represents the
transition formula obtained by replacing each atom $(i, q)$ in
$\varphi$ by the atom $(\stay, \la q, \la s, t_i\ra\ra)$\footnote{We
  assume directions in $\delta_1$ and $\delta_2$ are the numbers
  $1 (\leftt)$, $2 (\rightt)$, $3\ldots$, etc.}.

\subsubsection*{\bf Mode 5} Check values cannot be
reached. \label{sec:bf-mode-3.3} States drawn from:
\begin{align*}
  \mFive &\coloneqq \mFiveA\cup\mFiveB \\
  \mFiveA &\coloneqq D^k\times D\times \{\checkMiss\} \cup ((Q_1\sqcup
           Q_2)\times (D^k \times \{1,2\}\times \{\bot\}))\\
  \mFiveB &\coloneqq \big(D^k\times D\times \mathit{subterms}(P')\times \{\checkMiss\}\big)% \\&\qquad
  \cup \big((Q_1\sqcup
  Q_2)\times (D^k \times \{1,2\}\times\{\bot\}\times \mathit{subterms}(P'))\big)
\end{align*}
Transitions for $\mFiveA$:
\begin{itemize}[leftmargin=12pt]
\setlength\itemsep{0.5em}
\item $\delta(\la L, U, \checkMiss\ra, \_) =
  \bigwedge_{u\in U\cap Q_1}
  (\stay, \la u, \la L,1,\bot\ra\ra) \wedge \bigwedge_{u\in U\cap Q_2}
  (\stay, \la u, \la L,2,\bot\ra\ra)$
\item $\delta(\la u, \la L,m,\bot\ra\ra, x) = \adorn(\la
  L,m,\bot\ra, \dual(\delta_m(u, x))),\quad x\in\Delta$
\item $\delta(\la u, \la L,m,\bot\ra\ra, \rhs_{N_i}) = \tru \quad\,\text{if } u\notin L_i$
\item $\delta(\la u, \la L,m,\bot\ra\ra, \rhs_{N_i}) = \fals \quad
  \text{if } u\in L_i$
\end{itemize}
Transitions for $\mFiveB$:
\begin{itemize}[leftmargin=12pt]
\setlength\itemsep{0.5em}
\item $\delta(\la L, U, \alpha,\checkMiss\ra, \_) =
  \bigwedge_{u\in U\cap Q_1}
  (\stay, \la u, \la L,1,\bot,\alpha\ra\ra) \wedge \bigwedge_{u\in U\cap Q_2}
  (\stay, \la u, \la L,2,\bot,\alpha\ra\ra)$
\item $\delta(\la u, \la L,m,\bot,f(t_1,\ldots,t_r)\ra\ra, \_) =
  \adorn'(\la L,m,\bot,t_1,\ldots,t_r\ra, \dual(\delta_m(u, f)))$
\item $\delta(\la u, \la L,m,\bot,\rhs_{N_i}\ra\ra, \_) = \tru \quad\,\text{if } u\notin L_i$
\item $\delta(\la u, \la L,m,\bot,\rhs_{N_i}\ra\ra, \_) = \fals \quad
  \text{if } u\in L_i$
\end{itemize}

The notation $\dual(\varphi)$ represents a transition formula obtained
by replacing conjunction with disjunction and vice versa in the
(positive) Boolean formula $\varphi$.

\subsubsection*{\bf Mode 6} Check the first column is acceptable or
could still be acceptable later.\label{sec:bf-mode-2-1} States drawn
from:
\begin{align*}
\mSix\coloneqq D\times \{\verifySolution, \avoid\}.
\end{align*}

\begin{itemize}[leftmargin=12pt]
\setlength\itemsep{0.5em}
\item
  $\delta(\la U,\verifySolution\ra, \_) = \text{if}\,\, U\cap
  Q^i_1\neq\emptyset\,\, \text{then}\,\,\tru\,\,\text{else}\,\,\fals$
\item
  $\delta(\la U,\avoid\ra, \_) = \text{if}\,\, U\cap
  Q^i_2\neq\emptyset\,\, \text{then}\,\,\fals\,\,\text{else}\,\,\tru$
\end{itemize}

Any transition not described by the rules above has transition formula
$\fals$. The full set of states and the initial subset of states for
the automaton $A_I$ are
\begin{align*}
  Q &\coloneqq
      \mOne\sqcup\mTwo\sqcup\mThree\sqcup\mFour\sqcup\mFive\sqcup\mSix \qquad\text{and}\qquad
      Q_i\coloneqq\{\la \{\emptyset\}^k, \{\emptyset\}^k, \{Q_1\sqcup Q_2\}^k, \res\ra\}.
\end{align*}
By construction, we have the following.

\begin{lemma}
  $L(A_I) = \{t \in T_{\Gamma(\Delta, N)} \,:\,
  \solves(\extend(\dec(t),G'), I)\}$.
  \label{lemma:dsl-synthesis}
\end{lemma}
\begin{proof}
  \Cref{sec:proof-crefl-synth}.
\end{proof}

\subsection{Proof of \Cref{lemma:dsl-synthesis}}
\label{sec:proof-crefl-synth}

We use the notion of a run for two-way alternating tree automata
from~\cite{cachat-two-way}. Though our presentation in the main text
used terms over ranked alphabets, here it is more convenient to use
the language of labeled trees. Let $W$ be a set of directions using
which our terms over $\Gamma(\Delta,N)$ can be described as finite
$\Gamma(\Delta,N)$-labeled $W$-trees $(T,l)$. Let
$A_I=(Q, \Gamma(\Delta,N), I, \delta)$ and
$A_1=(Q_1,\Delta,F_1,\delta_1)$ and $A_2=(Q_2,\Delta,F_2,\delta_2)$.

First we argue that
$L(A_I)\subseteq \{t \in T_{\Gamma(\Delta, N)} \,:\,
\solves(\extend(\dec(t),G'), I)\}$. Suppose $t\in L(A_I)$, and let the
corresponding $W$-tree for $t$ be $(T,l)$. Let $(T_r, r)$ be an
accepting run of $A_I$ on $t$, which in our case means that each
branch is finite, and where $T_r$ is a $(W^*\times Q)$-labeled $Z$
tree, for a suitable set $Z$, and $r : T_r \rightarrow (W^*\times Q)$
labels each node of the run with $(p,q)$, whose meaning is that the
automaton passes through input tree node $p$ in state $q$.

The main observation is that from the run $(T_r,r)$ we can construct
table $T(\extend(\dec(t), G'))$ whose column corresponding to the
starting nonterminal is acceptable. This follows straightforwardly
from the construction of $A_I$, though it is tedious. We sketch the
argument now. The rows of the recursion table are the vectors $C$ in
the state component labeling the run at specific nodes. Along any path
in $T_r$, at the $i$th node at which a \modeone state is entered
(starting from root of $T_r$), including at the beginning of the run
when $i=1$, we have that the $i$th row $T(\extend(\dec(t), G'))[i:]$
is precisely the vector $C$. It is the case also that the automaton
state at that node has a vector $L$ which is the componentwise union
over all previous vectors $L$ at earlier positions in the run. Let
$Z_0=\la\emptyset, \ldots, \emptyset\ra\in \{\emptyset\}^k$, where
$k=|N|$. We claim that the state component $C$ satisfies
$C=H^{i-1}(Z_0)\setminus H^{i-2}(Z_0)$ and $L$ satisfies
$L=H^{i-2}(Z_0)$, for $i\ge 2$, where set operations are extended to
vectors by acting componentwise. The property for $L$ is preserved by
transitions in \modetwo, where the new value of $L$ is the union of
$C$ with the previous value of $L$, i.e., $L$ gets the value
$H^{i-2}(Z_0)\cup (H^{i-1}(Z_0)\setminus
H^{i-2}(Z_0))=H^{i-1}(Z_0)$. The property for $C$ is preserved by
transitions of \modethree, \modefour, and \modefive. In \modethree, it
is ensured that $C$ is drawn from remaining values in
$(\{Q_1\sqcup Q_2\}^k\setminus H^{i-2}(Z_0)$ and also that all of $C$
is in fact reachable, as the run cannot pass through the node labeled
$\prodend$ without having verified that all of $C$ is computable
either by existing productions in the tree $t$ or by productions from
$G'$, which are memorized in the states. In \modethree, it is also
ensured that the first column of $T(\extend(\dec(t), G'))$ is
acceptable by checking that the current row $C$ (copied in $W$) has as
its first component, either a set of values containing an initial
state of $A_1$, or if not, a set of values disjoint from the initial
states of $A_2$. In \modefour and \modefive, it is ensured that the
entries of $C$ are in fact generated by the productions available and
that no values excluded from $C$ can be generated (thus ensuring that
the guess for $C$ in fact contains all of
$H^{i-1}(Z_0)\setminus H^{i-1}(Z_0)$ rather than a strict
subset. Finally, since $T_r$ is finite and satisfies the transitions,
any time it returns to the root of $t$ we know that no
non-generalizing expression has been found, and furthermore, the run
cannot avoid ending in \modesix in a $\verifySolution$ state, which
ensures that the first column of $T(\extend(\dec(t), G'))$ is
acceptable.

Now we argue that
$L(A_I)\supseteq \{t \in T_{\Gamma(\Delta, N)} \,:\,
\solves(\extend(\dec(t),G'), I)\}$. Suppose $t$ is such that
$\solves(\extend(\dec(t),G'), I)\}$ holds. Let $N_i\in N$ be the
topmost nonterminal in the productions for $t$, which is the starting
nonterminal for $\dec(t)$ and also $\extend(\dec(t),G')$. Thus there
exists $e\in L(N_i)$ within $\extend(\dec(t),G')$ such that
$\solves(e,I)$ and for all non-generalizing $e'\in L(N_i)$ within
$\extend(\dec(t),G')$, we have
$\depth(e,\extend(\dec(t),G')) \le \depth(e', \extend(\dec(t), G'))$.
We can construct a run $r_t$ by making choices according to the
recursion table $T(\extend(\dec(t),G'))$. The run will only construct
the table up to stage $i$, where $i$ is the depth of a shallowest
parse tree $e_p$ for $e$ within $\extend(\dec(t),G')$, which by
assumption is at least as shallow as any parse tree for
non-generalizing $e'\in L(N_i)$ within $\extend(\dec(t),G')$. Thus the
run can satisfy the checks of \modesix within $i$ passes through a the
root node of $t$.

\section{Section 7}
\label{sec:adequate-macros-app}

\subsection{Adequate DSL Synthesis with Macros}
\label{sec:adeq-dsl-synth}

We describe an automaton $A_X$ which accepts an input tree if it
encodes a macro grammar which contains a solution for the set of
examples $X$. Because a macro can copy its input parameters, which
may, under outermost derivations, be arbitrary terms involving other
macros, our automaton will keep track of sets of states which the
parameters may evaluate to.

When evaluating a macro application, the automaton verifies that
\emph{sets} of states $Q'\subseteq Q_1$ are achievable, with $Q_1$ the
states for $A_1$. For example, when reading a macro $F(t_1, t_2)$,
$A_X$ nondeterministically guesses, for each parameter $i$, a subset
$Q_i'\subseteq Q_1$ that is achieved by $t_i$. It verifies these
guesses by descending into the arguments and recursively checking that
$t_i$ can evaluate to $Q_i'$, while simultaneously passing values
$Q_i'$ to all the productions of $F$. It then verifies that the $F$
productions together can produce the needed values $Q'$ given the
assumption about the arguments.

\textbf{Construction.} Suppose $A_{1}=(Q_{1},\Delta,I_1,\delta_1)$ is
an automaton accepting all expressions which satisfy the examples $X$.
We define a two-way alternating tree automaton
$A_X=(Q, \Gamma(\Delta, N), I, \delta)$. The automaton operates in two
modes, as before. In \textbf{mode 1}, it walks to the top spine of the
input tree in search of productions for a specific nonterminal. Having
found a production, it enters \textbf{mode 2}, in which it moves down
into the term corresponding to the right-hand side of the production,
simulating $A_1$ as it goes. Because nonterminals can be nested, these
two modes can be entered in a single transition.

Below we use $N_i\in N^{=0}$, $F\in N^{>0}$, $Y,Y'\in N$ with
$Y\neq N_i, Y'\neq F$, $q_i\in I_1$, $C\subseteq Q_1$,
$x, f\in\Delta$, and
$t_1,\ldots,t_r\in T_\Delta(\{\rhs_{N_i} \,:\, N_i\in N^{=0}\})$. We
use an underscore $"\_"$ to describe a default transition when no
other case matches. Let $a^*=\max(\{\arity(X) \,:\, X\in N\})$.

\subsubsection*{\bf Mode 1} Find productions. States are drawn from
 \label{sec:mode-1}
\begin{align*}
  \mOne\coloneqq \big(\Pp(Q_1)\times \{\start\}\big) \cup
  \big(\Pp(Q_1)\times N\big) \cup_{i\in [a^*]} \big( \Pp(Q_1)^i\times
  \Pp(Q_1)\times N^{=i} \big).
\end{align*}

\begin{align*}
  \delta(\la C, \start \ra, \route) = &\,\,(\down, \la C, \start\ra) \\
  \delta(\la C, \start \ra, \lhs_{N_i}) = &\,\,(\stay, \la C, N_i\ra) \\
  \delta(\la C,N_i\ra, \lhs_{Y}) = &\vee_{\left\{(C_1, C_2) : C_1\cup C_2=C\right\}}(\up, \la C_1,N_i\ra)\vee
                                     (\rightt, \la C_2, N_i\ra)\\
  \delta(\la C,N_i\ra, \lhs_{N_i}) = &\vee_{\left\{(C_1,C_2,C_3,C_4) :
                                       C_1\cup C_2\cup C_3\cup C_4=C\right\}}(\up, \la C_1,N_i\ra)\wedge (\leftt,
                                       C_2)
                                       \wedge (\rightt, \la C_3,
                                       N_i\ra)\\ & \wedge_{q\in C_4}\left(\vee_{(N_i, \,\alpha)\in P'}(\stay, \la
                                                   \{q\}, \alpha\ra)\right) \\
  \delta(\la C,N_i\ra, \_) = & \vee_{\{(C_1, C_2) : C_1\cup
                               C_2=C\}}(\up, \la
                               C_1,N_i\ra)\wedge_{q\in C_2}
                               \left(\vee_{(N_i,\, \alpha)\in P'}(\stay, \la \{q\},
                               \alpha\ra)\right)\\
  \delta(\la C_1,\ldots,C_k, C, F\ra, \lhs_{Y'}) = &\vee_{\left\{(C_1', C_2') : C_1'\cup C_2'=C\right\}}(\up, \la C_1,\ldots,C_k,C_1',F\ra)\wedge
                                                     (\rightt, \la C_1,\ldots,C_k,C_2', F\ra)\\
  \delta(\la C_1,\ldots,C_k, C, F\ra, \lhs_{F}) = &\vee_{\left\{(C_1',
                                                    C_2', C_3') :
                                                    C_1'\cup C_2'\cup
                                                    C_3'=C\right\}}\\ &(\up,
                                                                        \la
                                                                        C_1,\ldots,C_k,C_1',F\ra)\\
                                      & \wedge
                                        (\leftt, \la C_1,\ldots,C_k,C_2',F\ra)\wedge
                                        (\rightt, \la
                                        C_1,\ldots,C_k,C_3', F\ra)\\
  \delta(\la C_1,\ldots,C_k,C,F\ra, \_) = & (\up, \la C_1,\ldots,C_k,
                                        C,F\ra) \text{       // base language
                                            is macro-less}
\end{align*}

\subsubsection*{\bf Mode 2} Read productions. States drawn from \label{sec:bf-mode-2-1}
\begin{align*}
  &\mTwo\coloneqq \Pp(Q_1)\cup \big(\Pp(Q_1)\times \mathit{subterms}(P')\big),
  \\ &\text{ where }\mathit{subterms}(P') = \mathsmaller{\bigcup}_{(N_i,\, \alpha)\,\in\, P'\,} \mathit{subterms}(\alpha).
\end{align*}

\begin{align*}
  \delta(C, x) = &\wedge_{q\in C}\delta_1(q,x)\\
  \delta(\la C, f(t_1,\ldots, t_r)\ra, \_) = &\,\adorn(t_1,\ldots,t_r,
  \wedge_{q\in C}\,\delta_1(q, f))\\
  \delta(C, \rhs_{N_i}) = &\,(\stay, \la C, N_i\ra) \\
  \delta(\la C, \rhs_{N_i}\ra, \_) = &\,\vee_{\{(C_1,C_2) : C_1\cup
                                       C_2=C\}}(\stay, \la C_1,
                                       N_i\ra)\wedge
                                       \left(\wedge_{q\in C_2}\left(\vee_{(N_i,\, \alpha)\in P'}(\stay, \la \{q\},
                                       \alpha\ra)\right)\right)\\
  \delta(C, \rhs_{F}) = &\, \vee_{C_1,\ldots, C_k\in \Pp(Q_1)} \left(\wedge_{i\in [k]} (i,
                          C_i)\right)\wedge (\up, \la C_1,\ldots,C_k,C,F\ra)
\end{align*}

The notation $\adorn(t_1,\ldots, t_r, \varphi)$ represents a
transition formula obtained by replacing each atom of the form
$(i, q)$ in the Boolean formula $\varphi$ by the atom
$(\stay, \la q, t_i\ra)$.

Any transition not described by the rules above has transition formula
$\fals$. The full set of states and the initial states for the
automaton are
\begin{align*}
  Q \coloneqq \mOne\cup\mTwo, \qquad I=\{\la \{q_i\}, \start\ra \,:\, q_i\in I_1\}\subseteq \mOne.
\end{align*}
We have the following by construction.

\begin{lemma}
  $L(A_X) = \big\{t \in T_{\Gamma(\Delta, N)} \,:\,
  \solves(\extend(\dec(t),G'), X)\big\}$.
  \label{lemma:grammar-synthesis-app}
\end{lemma}

The rest of the proof is similar to that of adequate DSL synthesis for
grammars and gives us:

\begin{theorem}
  Adequate DSL synthesis with macros is decidable for any language
  whose semantics over fixed structures can be evaluated by tree
  automata. Furthermore, the set of solutions corresponds to a regular
  set of trees.  \label{thm:adequate-macros-app}
\end{theorem}

The size of the two-way automaton $A_X$ from this section is
exponential in the size of $A_1$, giving us the following, similar to
before.

\begin{corollary}
  For the languages covered in \Cref{thm:adequate-macros-app},
  adequate DSL synthesis with macros is decidable in time
  $\emph{poly}(|\metagram|)\cdot \exp(l\cdot \exp(m)),$ where $l$ is
  the number of instances and $m$ is the maximum size over all
  instance automata.
\end{corollary}

\subsection{DSL synthesis with macros}
\label{sec:section-7-dsl}

The \emph{macro depth} of a tree $t\in T_{\Delta\cup N}(\Nat)$ is the
maximum number of macro symbols encountered along any root-to-leaf
path in $t$. For example, $f(a)$, $N(a)$, and $N(1, H(2))$ have macro
depths $0$, $1$, and $2$, respectively. A macro grammar \emph{has
  macro depth bounded by $b$} if, for all of its rules $(N,\alpha)$,
$\alpha$ has macro depth no more than $b$. Note that given a macro
grammar $G$ and bound $b\in\Nat$, we can easily verify that $G$ has
macro depth bounded by $b$.

A meta grammar $\metagram$ has macro depth bounded by $b\in\Nat$ if,
for every $G\in L(\metagram)$, $G$ has macro depth bounded by
$b$. Given a meta grammar $\metagram$, we can compute the minimum $b$
(if it exists) for which it has macro depth bounded by $b$. We can do
it by checking for ``cycles'' in the rules corresponding to right-hand
sides of object grammar rules which contain macro symbols along
them. We can keep track of the maximum macro depth possible while
looking for cycles.  If such a cycle is found then there is no bound
$b$ and if not we output the maximum macro depth as the bound $b$.

The proof of \Cref{thm:order-macro-theorem} is similar to
\Cref{sec:decidability}. We construct an automaton $A_I$ which accepts
an input tree if it encodes a macro grammar which solves an instance
$I=(X,Y)$ for parse tree depth ordering. Because macro grammars allow
applications of macros to be nested, and because this problem requires
paying attention to the depth of expressions, the automaton $A_I$ uses
many more states than the construction from \Cref{sec:decidability} in
order to handle nested macros. We describe a construction of $A_I$,
which is parameterized by a depth bound $b\in\Nat$, such that $A_I$
operates as expected over grammars with macro nesting depth bounded by
$b$.

The construction shares much of the structure related to recursion
tables, but in building up each row one after the other, the automaton
must keep fine-grained information about a bounded number of previous
rows to handle boundedly-nested macros. Additionally, the entries of
the recursion table corresponding to macro symbols indicate
\emph{functions} from tuples of value sets to a set of values achieved
by the macro symbol with a given depth budget. In fact, the entries
for macro symbols indicate functions from tuples of sets of values to
the \emph{new values} that are achievable given higher depth budget
(as in the construction of \Cref{sec:decidability}).

The main complication, as mentioned above, is that to keep track of
depth in the presence of macros, we need the automaton to make a
distinction between values achieved at different previous depths, as
opposed to lumping them together as a set of values achievable in
depth less than some bound. The way this can be handled is by allowing
the automaton to encode in its states several previous rows of the
recursion table up to each individual depth. This is sufficient to
implement the same protocol for nondeterministically guessing and
verifying the rows of the recursion table for a macro grammar.

\subsubsection{Construction of $A_I$
}\label{sec:dsl-construction-macro} As before we
assume non-deterministic top-down instance automata
$A_{1}=(Q_{1},\Delta,Q^i_1,\delta_1)$ and
$A_{2}=(Q_{2},\Delta,Q^i_2,\delta_{2})$. The automaton $A_I$ will
guess the construction of tables similar to the recursion tables of
\Cref{sec:decidability}. A difference is that the entries for
\emph{non-macro nonterminals} range over the powerset
$\mathcal{P}(Q_1\sqcup Q_2)$ and the entries for \emph{macro
  nonterminals} range over functions over these sets. Let us use
$D\coloneqq\mathcal{P}(Q_1\sqcup Q_2)$ as shorthand in the
remainder. Fix a macro nesting depth bound $b\in\Nat$. Consider a
macro symbol $F^1$. Entries for its column in the table range over
functions $[D\rightarrow D]$. For a symbol $K^2$ they range over
$[D^2\rightarrow D]$, etc. For any macro grammar $G$, the intuition is
that if $T(G)[i,j] = f\in [D\rightarrow D]$, then, provided an
argument which can generate values $v\in D$, the nonterminal $N_j$ can
generate terms evaluating to values $f(v)$ using derivations of depth
$i$ and no smaller.

We now define the two-way alternating tree automaton
$A_I=(Q, \Gamma(\Delta, N), Q_i, \delta)$ with acceptance defined by
the existence of a finite run satisfying the transition formulas. We
describe the states $Q$ and their transition formulas grouped by
functionality. The transitions are similar to those of
\Cref{sec:decidability} and are organized similarly as well. We note
salient differences alongside the transitions.  We assume there are
$k$ nonterminal symbols. Below we use $i,j\in[k]$, $m\in\{1,2\}$,
$u\in Q_1\sqcup Q_2$, $u_1\in Q_1$, $u_2\in Q_2$, $U,V\in D$,
$L, R, R', C, C', W\in D^k$, $N_i, N_j\in N$, $f\in \Delta^r$, and
$t_1,\ldots,t_r\in T_\Delta(\{\rhs_{N_i} \,:\, N_i\in N\})$. We use an
underscore $"\_"$ to describe a default transition when no other case
matches.

We assume a macro nesting depth bound of $b\in\Nat$. We assume table
rows are drawn from a set $\Rows$. We assume that the entries of rows
which correspond to macro symbols contain vectors representing the
values the functions take for some fixed ordering of the elements of
their finite input sets. For instance, for a binary macro symbol
$G(1,2)$, its entries in the table correspond to functions of type
$D\times D\rightarrow D$, but are manipulated in the construction as
vectors $v\in D^l$, where $l=|D\times D|$. We write $\argsType$ to
mean the domain of vectors of arguments for macro symbols, e.g. it
includes $D\times D$ for binary macro symbols, etc. We write
$\dropFirst$ to mean dropping the earliest (first) row in a sequence
of rows. We write $\emptyRow$ to mean the row whose entries are all
empty; for a vector corresponding to a function this means a vector of
empty sets. We write $\emptyset$ to mean either an empty set or a
vector of empty sets, depending on the context. We write $\fullRow$ to
mean the row whose entries are all full; for sets this is
$Q_1\sqcup Q_2$ and for vectors over such sets which model functions
this is $(Q_1\sqcup Q_2)^l$, for an appropriate vector length $l$
depending on the number of parameters for the macro symbol at a
particular index in the row. In some cases we use $\cup$ to mean both
union of sets and componentwise union for vectors over sets, depending
on context (e.g. first bullet of \Cref{sec:bf-mode-3.1-macro}).

\subsubsection*{\bf Mode 1} Reset to the top of the input tree.
\label{sec:mode-1-macro} States are drawn from:
\begin{align*}
\mOne\coloneqq \Rows^{b+3}\times \{\res, \guess\}.
\end{align*}

\begin{itemize}[leftmargin=12pt]
\setlength\itemsep{0.5em}
\item $\delta(\la \rows,L,C,R, \res\ra, \route) = (\down, \la \rows,L,C,R, \guess\ra)$
\item $\delta(\la \rows,L,C,R, \res\ra, \_) = (\up, \la \rows,L,C,R, \res\ra)$
\item $\delta(\la \rows,L,C,R,\guess\ra, \lhs_{N_i}) = (\stay, \la \rows,L,C,R,i,\guessRow\ra)$
\end{itemize}

\subsubsection*{\bf Mode 2} Guess next row of the recursion
table.\label{sec:bf-mode-3-macro} States drawn from:
\begin{align*}
\mTwo\coloneqq \Rows^{b+3}\times[k]\times \{\guessRow\}.
\end{align*}

\begin{itemize}[leftmargin=12pt]
\item[] // Forget the deepest row ($\dropFirst$) of the table as we
  add a new one
\item $\delta(\la \rows,L,C,R,i,\guessRow\ra, \_) = \bigvee_{(C', R')\in\, \okay(R)}
    (\stay, \la \dropFirst(\rows), L, L\cup C, C', C', R', i,\guessProd\ra)$
  \item[] with $\okay(R) \coloneqq \big\{(C',R') \in \Rows\times\Rows
      \,:\,
      C'\cup R'=R,\, C'\neq \emptyRow\big\}$
\end{itemize}

\subsubsection*{\bf Mode 3} Guess the contributions of productions to
each row entry. States drawn from: \label{sec:bf-mode-3.1-macro}
\begin{align*}
\mThree\coloneqq \Rows^{b+4}\times[k]\times \{\guessProd\}.
\end{align*}

\begin{itemize}[leftmargin=12pt]
\setlength\itemsep{0.5em}
\item
  $\delta(\la \rows, L,C,W,R,i, \guessProd\ra, \lhs_{N_j^{0}}) =
   \bigvee_{\left\{(U, V) \,:\, U\cup V = C_j\right\}} (\leftt, \la
  \rows, L, \noArgs, U, \checkHit\ra) \wedge\\ (\leftt, \la \rows, L,R_j, \checkMiss\ra) \wedge\\
  (\rightt, \la\rows, L, \la C_1,\ldots,C_{j-1},V,\ldots,C_k\ra, W, R,i, \guessProd\ra)$
\item
  $\text{// Compute the macro arguments using $\args(j,\argIdx)$ to give to the child state}\\$
  $\delta(\la \rows, L,C,W,R,i, \guessProd\ra, \lhs_{N_j^{>0}}) =
  \bigvee_{\left\{(U, V) \,:\, U\cup V = C_j\right\}}
  \\\left(\bigwedge_{\argIdx\in [|C_j|]}(\leftt, \la
  \rows, L, \args(j,\argIdx), U, \checkHit\ra)\right) \wedge (\leftt, \la \rows, L,R_j, \checkMiss\ra) \wedge\\
  (\rightt, \la\rows, L, \la C_1,\ldots,C_{j-1},V,\ldots,C_k\ra, W, R,i, \guessProd\ra)$
\item $\delta(\la \rows, L, C, W, R,i, \guessProd\ra, \prodend) =\\
  \text{ if }\,\,\exists (N_j\in N\setminus N').\, C_j\neq \emptyset\\
  \text{ then}\,\,\fals \\ \text{ else} \,\,// \,\,\text{Note: }\, i \text{ indexes the starting nonterminal}$
\item[]\hspace{0.2in}
  \begin{tabular}[t]{l}
    $\Big((\stay, \la L_i\cup W_i, \verifySolution\ra)$
    $\vee \big((\stay, \la \rows, L, W, R, \res\ra)\wedge  (\stay, \la L_i\cup W_i,
  \avoid\ra)\big)\Big)$ \\
  $\wedge\Big(\bigwedge_{N_j\in
    (N')^0}\bigwedge_{u\in C_j}\bigvee_{(N_j,\alpha)\in P'}(\stay, \la\rows, L, \noArgs,\{u\},
    \alpha,\checkHit\ra)\Big)$ \\
  $\wedge\Big(\bigwedge_{N_j\in
    (N')^{>0}}\bigwedge_{\argIdx\in [|C_j|]}\bigwedge_{u\in C_j(\argIdx)}\bigvee_{(N_j,\alpha)\in P'}(\stay, \la\rows, L, \args(j,\argIdx),\{u\},
    \alpha,\checkHit\ra)\Big)$ \\
  $\wedge\Big(\bigwedge_{N_j\in N'}\bigwedge_{(N_j,\alpha)\in P'}(\stay, \la\rows, L, R_j,
    \alpha,\checkMiss\ra)\Big)$
  \end{tabular}
\end{itemize}

\subsubsection*{\bf Mode 4} Check a set of values can be
reached. \label{sec:bf-mode-3.2-macro} States drawn from:

\begin{align*}
  \mFour &\coloneqq \mFourA\cup \mFourB \\
  \mFourA &\coloneqq \Rows^b\times \Rows\times \argsType\times D\times
            \{\checkHit\} \\ &\cup ((Q_1\sqcup
  Q_2)\times (\Rows^{b+1}\times\argsType\times \{0,\ldots,b-1\} \times \{1,2\})) \\
  \mFourB &\coloneqq \big(\Rows^b\times\Rows\times \argsType\times D\times
            \mathit{subterms}(P')\times \{\checkHit\}\big) %\\ &\qquad
                                                                \\ &\cup \big((Q_1\sqcup
  Q_2)\times (\Rows^{b+1}\times\argsType\times \{0,\ldots,b-1\} \times \{1,2\}\times \mathit{subterms}(P'))\big),\\
  &\text{where }\mathit{subterms}(P') = \mathsmaller{\bigcup}_{(N_i,\,\alpha)\,\in\, P'\,} \mathit{subterms}(\alpha)
\end{align*}
Transitions for $\mFourA$:
\begin{itemize}[leftmargin=12pt]
\setlength\itemsep{0.5em}
\item
    $\text{// Current macro nesting depth starts at $0$}\\$
    $\delta(\la \rows,L,\topArgs,U, \checkHit\ra, \_) = \bigwedge_{u_1\in U\cap Q_1}
  (\stay, \la u_1, \la \rows,L,\topArgs,\currDepth=0, 1\ra \ra) \wedge \\\bigwedge_{u_2\in U\cap Q_2}
  (\stay, \la u_2, \la \rows,L,\topArgs,\currDepth=0, 2\ra \ra)$
\item $\delta(\la u, \la \rows,L,\topArgs,\currDepth, m\ra\ra, x) =\\ \adorn(\la \rows,L,\topArgs,\currDepth, m\ra,
  \delta_{m}(u, x)),\quad x\in\Delta$
\item
    $\text{// Guess how much available depth to budget for arguments and how
       much for macro expansion}\\$
  $\delta(\la u, \la \rows,L,\topArgs,\currDepth, m\ra\ra, \rhs_{N_i^{>0}}) =
  \bigvee_{\offset\in [b-\currDepth]}\bigvee_{\{\argIdx\in
    [|target|]\,:\, u\in target(\argIdx)\}}\\
  \bigwedge_{p\in [|\args(i,\argIdx)|]}\bigwedge_{u'\in
    \args(i,\argIdx)(p)}(p, \la u', \la
  \rows,L,\topArgs,\currDepth+\offset,m \ra\ra)\\
  \text{// Do not need to check the arguments evaluate precisely to
    chosen ones}$
  \begin{itemize}
  \item[] where $target=\la\rows,L\ra(b-\currDepth)(i)$
  \end{itemize}
\item $\delta(\la u, \la \rows,L,\topArgs,\currDepth, m\ra\ra, \rhs_{N_i^0}) = \tru \quad\, \text{if
  } u\in target\\$
  $\delta(\la u, \la \rows,L,\topArgs,\currDepth, m\ra\ra, \rhs_{N_i^0}) = \fals \quad \text{if
  } u\notin target$
  \begin{itemize}
  \item[] where $target=\la\rows,L\ra(b-\currDepth)(i)$
  \end{itemize}
\item $\delta(\la u, \la \rows,L,\topArgs,\currDepth, m\ra\ra, i) =
  \tru\quad\,\text{if } u\in \topArgs(i)\\$
 $\delta(\la u, \la \rows,L,\topArgs,\currDepth, m\ra\ra, i) =
  \fals\quad\,\text{if } u\notin \topArgs(i)$
\item[] $\text{//} \,\,i\,\, \text{is a formal macro parameter}$
\end{itemize}

Transitions for $\mFourB$:
\begin{itemize}[leftmargin=12pt]
\setlength\itemsep{0.5em}
\item $\delta(\la \rows,L,\topArgs,U,\alpha, \checkHit\ra, \_) = \bigwedge_{u_1\in U\cap Q_1}
  (\stay, \la u_1, \la \rows,L,\topArgs,\currDepth=0, 1, \alpha\ra \ra) \wedge \\\bigwedge_{u_2\in U\cap Q_2}
  (\stay, \la u_2, \la \rows,L,\topArgs,\currDepth=0, 2, \alpha\ra \ra)$
\item $\delta(\la u, \la \rows,L,\topArgs,\currDepth, m, f(t_1,\ldots,t_r)\ra\ra, \_) =
  \\\adorn'(\la \la \rows,L,\topArgs,\currDepth, m\ra, t_1,\ldots, t_r\ra,
  \delta_{m}(u, f))$
\item $\delta(\la u, \la \rows,L,\topArgs,\currDepth, m,\rhs_{N_i^{>0}}\ra\ra, \_) =
  \bigvee_{\offset\in [b-\currDepth]}\bigvee_{\{\argIdx\in
    [|target|]\,:\, u\in target(\argIdx)\}}\\
  \bigwedge_{p\in [|\args(i,\argIdx)|]}\bigwedge_{u'\in
    \args(i,\argIdx)(p)}(p, \la u', \la
  \rows,L,\topArgs,\currDepth+\offset,m \ra\ra)\\
  \text{// Do not need to check the arguments evaluate precisely to
    chosen ones}$
  \begin{itemize}
  \item[] where $target=\la\rows,L\ra(b-\currDepth)(i)$
  \end{itemize}
\item $\delta(\la u, \la \rows,L,\topArgs, \currDepth,m, \rhs_{N_i^0}\ra\ra, \_) = \tru \quad\, \text{if
  } u\in target\\
\delta(\la u, \la \rows,L,\topArgs,\currDepth, m, \rhs_{N_i^0}\ra\ra, \_) = \fals \quad \text{if
  } u\notin target$
  \begin{itemize}
  \item[] where $target=\la\rows,L\ra(b-\currDepth)(i)$
  \end{itemize}
\item $\delta(\la u, \la \rows,L,\topArgs,\currDepth, m,i\ra\ra, \_) =
  \tru\quad\,\text{if } u\in \topArgs(i)\\$
  $\delta(\la u, \la \rows,L,\topArgs,\currDepth, m,i\ra\ra, \_) =
  \fals\quad\,\text{if } u\notin \topArgs(i)$
\item[] $\text{//} \,\,i\,\, \text{is a formal macro parameter}$
\end{itemize}

The notation $\adorn(s, \varphi)$ represents the transition formula
obtained by replacing each atom $(i, q)$ in the Boolean formula
$\varphi$ by the atom $(i, \la q, s\ra)$. The notation
$\adorn'(\la s, t_1,\ldots, t_r\ra, \varphi)$ represents the
transition formula obtained by replacing each atom $(i, q)$ in
$\varphi$ by the atom $(\stay, \la q, \la s, t_i\ra\ra)$\footnote{We
  assume directions in $\delta_1$ and $\delta_2$ are the numbers
  $1 (\leftt)$, $2 (\rightt)$, $3\ldots$, etc.}.

\subsubsection*{\bf Mode 5} Check values cannot be
reached. \label{sec:bf-mode-3.3-macro} States drawn from:

\begin{align*}
  \mFive &\coloneqq \mFiveA\cup \mFiveB \\
  \mFiveA &\coloneqq \Rows^b\times \Rows\times \argsType\times D\times
            \{\checkMiss\} \\ &\cup ((Q_1\sqcup
  Q_2)\times (\Rows^{b+1}\times\argsType\times \{0,\ldots,b-1\} \times \{1,2\}\times\{\bot\})) \\
  \mFiveB &\coloneqq \big(\Rows^b\times\Rows\times \argsType\times D\times
            \mathit{subterms}(P')\times \{\checkMiss\}\big) %\\ &\qquad
                                                                \\ &\cup \big((Q_1\sqcup
  Q_2)\times (\Rows^{b+1}\times\argsType\times \{0,\ldots,b-1\} \times \{1,2\}\times\{\bot\}\times \mathit{subterms}(P'))\big),\\
  &\text{where }\mathit{subterms}(P') = \mathsmaller{\bigcup}_{(N_i,\,\alpha)\,\in\, P'\,} \mathit{subterms}(\alpha)
\end{align*}
Transitions for $\mFiveA$:
\begin{itemize}[leftmargin=12pt]
\setlength\itemsep{0.5em}
\item $\delta(\la \rows,L,\topArgs,U, \checkMiss\ra, \_) = \bigwedge_{u_1\in U\cap Q_1}
  (\stay, \la u_1, \la \rows,L,\topArgs,\currDepth=0, 1,\bot\ra \ra) \wedge \\\bigwedge_{u_2\in U\cap Q_2}
  (\stay, \la u_2, \la \rows,L,\topArgs,\currDepth=0, 2,\bot\ra \ra)$
\item $\delta(\la u, \la \rows,L,\topArgs,\currDepth, m,\bot\ra\ra, x) =\\ \adorn(\la \rows,L,\topArgs,\currDepth, m,\bot\ra,
  \delta_{m}(u, x)),\quad x\in\Delta$
\item $\delta(\la u, \la \rows,L,\topArgs,\currDepth, m,\bot\ra\ra, \rhs_{N_i^{>0}}) =
  \bigwedge_{\offset\in [b-\currDepth]}\bigwedge_{\{\argIdx\in
    [|target|]\,:\, u\in target(\argIdx)\}}\\
  \bigvee_{p\in [|\args(i,\argIdx)|]}\bigvee_{u'\in
    \args(i,\argIdx)(p)}(p, \la u', \la
  \rows,L,\topArgs,\currDepth+\offset,m,\bot \ra\ra)$
  \begin{itemize}
  \item[] where $target=\la\rows,L\ra(b-\currDepth)(i)$
  \end{itemize}
\item $\delta(\la u, \la \rows,L,\topArgs,\currDepth, m,\bot\ra\ra, \rhs_{N_i^0}) = \fals \quad\, \text{if
  } u\in target\\$
  $\delta(\la u, \la \rows,L,\topArgs,\currDepth, m,\bot\ra\ra, \rhs_{N_i^0}) = \tru \quad \text{if
  } u\notin target$
  \begin{itemize}
  \item[] where $target=\la\rows,L\ra(b-\currDepth)(i)$
  \end{itemize}
\item $\delta(\la u, \la \rows,L,\topArgs,\currDepth, m,\bot\ra\ra, i) =
  \fals\quad\,\text{if } u\in \topArgs(i)\\$
 $\delta(\la u, \la \rows,L,\topArgs,\currDepth, m,\bot\ra\ra, i) =
  \tru\quad\,\text{if } u\notin \topArgs(i)$
\item[] $\text{//} \,\,i\,\, \text{is a formal macro parameter}$
\end{itemize}

Transitions for $\mFiveB$:
\begin{itemize}[leftmargin=12pt]
\setlength\itemsep{0.5em}
\item $\delta(\la \rows,L,\topArgs,U,\alpha, \checkMiss\ra, \_) = \bigwedge_{u_1\in U\cap Q_1}
  (\stay, \la u_1, \la \rows,L,\topArgs,\currDepth=0, 1,\bot, \alpha\ra \ra) \wedge \bigwedge_{u_2\in U\cap Q_2}
  (\stay, \la u_2, \la \rows,L,\topArgs,\currDepth=0, 2,\bot,\alpha\ra \ra)$
\item $\delta(\la u, \la \rows,L,\topArgs,\currDepth, m,\bot, f(t_1,\ldots,t_r)\ra\ra, \_) =
  \\\adorn'(\la \la \rows,L,\topArgs,\currDepth, m,\bot\ra, t_1,\ldots, t_r\ra,
  \delta_{m}(u, f))$
\item $\delta(\la u, \la \rows,L,\topArgs,\currDepth, m,\bot,\rhs_{N_i^{>0}}\ra\ra, \_) =
  \\\bigwedge_{\offset\in [b-\currDepth]}\bigwedge_{\{\argIdx\in
    [|target|]\,:\, u\in target(\argIdx)\}}\\
  \bigvee_{p\in [|\args(i,\argIdx)|]}\bigvee_{u'\in
    \args(i,\argIdx)(p)}(p, \la u', \la
  \rows,L,\topArgs,\currDepth+\offset,m,\bot \ra\ra)$
  \begin{itemize}
  \item[] where $target=\la\rows,L\ra(b-\currDepth)(i)$
  \end{itemize}
\item $\delta(\la u, \la \rows,L,\topArgs, \currDepth,m,\bot, \rhs_{N_i^0}\ra\ra, \_) = \fals \quad\, \text{if
  } u\in target\\
\delta(\la u, \la \rows,L,\topArgs,\currDepth, m,\bot, \rhs_{N_i^0}\ra\ra, \_) = \tru \quad \text{if
  } u\notin target$
  \begin{itemize}
  \item[] where $target=\la\rows,L\ra(b-\currDepth)(i)$
  \end{itemize}
\item $\delta(\la u, \la \rows,L,\topArgs,\currDepth, m,\bot,i\ra\ra, \_) =
  \fals\quad\,\text{if } u\in \topArgs(i)\\$
  $\delta(\la u, \la \rows,L,\topArgs,\currDepth, m,\bot,i\ra\ra, \_) =
  \tru\quad\,\text{if } u\notin \topArgs(i)$
\item[] $\text{//} \,\,i\,\, \text{is a formal macro parameter}$
\end{itemize}

The notation $\dual(\varphi)$ represents a transition formula obtained
by replacing conjunction with disjunction and vice versa in the
(positive) Boolean formula $\varphi$.

\subsubsection*{\bf Mode 6} Check the first column is acceptable or
could still be acceptable later.\label{sec:bf-mode-2-1-macro} States drawn
from:
\begin{align*}
\mSix\coloneqq D\times \{\verifySolution, \avoid\}.
\end{align*}

\begin{itemize}[leftmargin=12pt]
\setlength\itemsep{0.5em}
\item
  $\delta(\la U,\verifySolution\ra, \_) = \text{if}\,\, U\cap
  Q^i_1\neq\emptyset\,\, \text{then}\,\,\tru\,\,\text{else}\,\,\fals$
\item
  $\delta(\la U,\avoid\ra, \_) = \text{if}\,\, U\cap
  Q^i_2\neq\emptyset\,\, \text{then}\,\,\fals\,\,\text{else}\,\,\tru$
\end{itemize}

Any transition not described by the rules above has transition formula
$\fals$. The full set of states and the initial subset of states for
the automaton $A_I$ are
\begin{align*}
Q &\coloneqq
    \mOne\sqcup\mTwo\sqcup\mThree\sqcup\mFour\sqcup\mFive\sqcup\mSix \qquad\text{and}\qquad
 Q_i\coloneqq\{\la \{\emptyRow\}^{b+2}, \fullRow, \res\ra\}.
\end{align*}
By construction we have the following.

\begin{lemma}
  $L(A_I) = \{t \in T_{\Gamma(\Delta, N)} \,:\,
  \solves(\extend(\dec(t),G'), I)\}$.
  \label{lemma:dsl-synthesis-macro}
\end{lemma}
\begin{proof}
  Similar argument to \Cref{sec:proof-crefl-synth}.
\end{proof}

The rest of the proof involves standard constructions essentially
identical to those of \Cref{sec:decidability}.